\documentclass[reqno,12pt]{amsart}

\usepackage{color} 
\usepackage{ifpdf}
\ifpdf
    \usepackage[pdftex]{graphicx}
    \usepackage[pdftex]{hyperref}
    \hypersetup{
        unicode=false,          
        pdftoolbar=true,        
        pdfmenubar=true,        
        pdffitwindow=false,     
        pdfstartview={FitH},    
        pdftitle={MCP Article},      
        pdfauthor={Michael Holst},   
        pdfsubject={Mathematics},    
        pdfcreator={Michael Holst},  
        pdfproducer={Michael Holst}, 
        pdfkeywords={PDE, analysis, mathematical physics}, 
        pdfnewwindow=true,      
        colorlinks=true,        
        linkcolor=red,          
        citecolor=blue,         
        filecolor=magenta,      
        urlcolor=cyan           
    }

\else
    \usepackage{graphicx}
    \usepackage{pstricks}
    
    \newcommand{\href}[2]{#2}
\fi

\usepackage{times}
\usepackage{amsfonts}

\usepackage{amsmath}
\usepackage{amsthm}
\usepackage{amssymb}
\usepackage{amsbsy}
\usepackage{amscd}

\usepackage{enumerate}
\usepackage{verbatim}
\usepackage{subfigure}

\newtheorem{theorem}{Theorem}[section]
\newtheorem{corollary}[theorem]{Corollary}
\newtheorem{lemma}[theorem]{Lemma}

\newtheorem{proposition}[theorem]{Proposition}

\newtheorem{remark}[theorem]{Remark}

\numberwithin{equation}{section}  

  \newcounter{mnote}
  \let\oldmarginpar\marginpar
    \renewcommand\marginpar[1]{\-\oldmarginpar[\raggedleft\footnotesize #1]%
    {\raggedright\footnotesize #1}}

\definecolor{myblue}{rgb}{0.2,0.2,0.7}
\definecolor{mygreen}{rgb}{0,0.6,0}
\definecolor{mycyan}{rgb}{0,0.6,0.6}
\definecolor{myred}{rgb}{0.9,0.2,0.2}
\definecolor{mymagenta}{rgb}{0.9,0.2,0.9}
\definecolor{mywhite}{rgb}{1.0,1.0,1.0}
\definecolor{myblack}{rgb}{0.0,0.0,0.0}

\def\mathbi#1{\textbf{\em #1}}
\newcommand{\beq}{\begin{equation}}
\newcommand{\eeq}{\end{equation}}
\newcommand{\beqa}{\begin{eqnarray}}
\newcommand{\eeqa}{\end{eqnarray}}
\newcommand{\leqs}{\leqslant}      
     
\newcommand{\geqs}{\geqslant}      
      
\renewcommand{\div}{{\operatorname{div}}}

\newcommand{\tiwedge}{\mbox{{\tiny $\wedge$}}}
\newcommand{\tivee}{\mbox{{\tiny $\vee$}}}

\newcommand{\cH}{{\mathcal H}}

\newcommand{\cL}{{\mathcal L}}
\newcommand{\cM}{{\mathcal M}}
\newcommand{\cN}{{\mathcal N}}

\newcommand{\cP}{{\mathcal P}}

\newcommand{\cS}{{\mathcal S}}

\newcommand{\cY}{{\mathcal Y}}

\newcommand{\ttk}{{\tt k}}

\newcommand{\bV}{{\bf V}}
\newcommand{\bW}{{\bf W}}
\newcommand{\bX}{{\bf X}}

\newcommand{\bx}{{\bf x}}

\def\mathbi#1{\textbf{\em #1}}

\newcommand{\biW}{\mathbi{W\,}}

\newcommand{\biw}{\mathbi{w}}

\newcommand{\e}{\epsilon}
\begin{document}

\title[ Non-CMC Solutions on Asymptotically Euclidean Manifolds with Boundary ]{ Non-CMC Solutions to the Einstein Constraint Equations on Asymptotically Euclidean Manifolds with Apparent Horizon Boundaries  }

\author[M. Holst]{Michael Holst}
\email{mholst@math.ucsd.edu}

\author[C. Meier]{Caleb Meier}
\email{c1meier@math.ucsd.edu}

\address{Department of Mathematics\\
         University of California San Diego\\ 
         La Jolla CA 92093}

\thanks{MH was supported in part by 
        NSF Awards~1065972, 1217175, and 1262982.}
\thanks{CM was supported in part by NSF Award~1065972.}

\date{\today}


\begin{abstract}
In this article we further develop the solution theory for the Einstein 
constraint equations on an $n$-dimensional, asymptotically Euclidean 
manifold $\cM$ with interior boundary $\Sigma$.
Building on recent results for both the asymptotically Euclidean and compact
with boundary settings, we show existence of far-from-CMC and near-CMC 
solutions to the conformal formulation of the Einstein constraints when
nonlinear Robin boundary conditions are imposed on $\Sigma$, similar to 
those analyzed previously by Dain (2004), by Maxwell (2004, 2005), 
and by Holst and Tsogtgerel (2013) as a model of black holes in various CMC 
settings, and by Holst, Meier, and Tsogtgerel (2013) in the setting of 
far-from-CMC solutions on compact manifolds with boundary.
These ``marginally trapped surface'' Robin conditions ensure that the 
expansion scalars along null geodesics perpendicular to the boundary region 
$\Sigma$ are non-positive, which is considered the correct mathematical
model for black holes in the context of the Einstein constraint equations.
Assuming a suitable form of weak cosmic censorship, the results 
presented in this article guarantee the existence of initial data 
that will evolve into a space-time containing an arbitrary number of 
black holes.
A particularly important feature of our results are the minimal restrictions 
we place on the mean curvature, giving both near- and far-from-CMC results 
that are new. 
\end{abstract}

\maketitle


\vspace*{-0.5cm}
\tableofcontents
\vspace*{-0.5cm}

\section{Introduction}
\label{sec:intro}
In this paper we consider the Einstein constraint equations on an $n$-dimensional, asymptotically Euclidean manifold
$\cM$ with boundary $\Sigma$.  Using the recent work in \cite{DIMM13,CHT13,HoTs10a}, we show
that far-from-CMC and near-CMC solutions exist to the conformal formulation of the Einstein constraints
when nonlinear Robin boundary conditions are imposed on $\Sigma$ similar to those developed in
\cite{SD04,dM05a,HoTs10a}.  These ``marginally trapped surface", Robin conditions ensure that
the expansion scalars along null geodesics perpendicular to the boundary region $\Sigma$ are non-positive.  Therefore,
assuming a suitable form of weak cosmic censorship, the results presented here provide a method
to construct initial data that will evolve into a space-time containing an arbitrary number of black holes.
Moreover, this method imposes very few restrictions on the mean curvature. 

We recall that the Einstein constraint equations on a given manifold $\cM$ take the form
\begin{align}
&\hat{R}-\hat{K}^{ab}\hat{K}_{ab}+\hat{K} =\hat{\rho},\label{eq1:27nov13}  \\
&\hat{D}^a\hat{K}-\hat{D}_b\hat{K}^{ab} = -\hat{J}^a, \label{eq2:27nov13}
\end{align}
where \eqref{eq1:27nov13} is the Hamiltonian constraint and \eqref{eq2:27nov13}
is the momentum constraint.  In the above system, $\hat{R}$ and $\hat{D}$ are
the scalar curvature and connection with respect to the metric $\hat{g}_{ab}$, and
$\hat{K}_{ab}$ and $\hat{K}$ are the extrinsic curvature tensor and its trace.
The above underdetermined system
imposes conditions on initial data $(\cM,\hat{g}_{ab},\hat{K}_{ab})$
for the initial value formulation of Einstein's equation.

In order to obtain solutions to \eqref{eq1:27nov13}-\eqref{eq2:27nov13} satisfying the marginally trapped surface conditions,
we impose boundary conditions on $(\hat{g}_{ab},\hat{K}_{ab})$ over $\Sigma$.
Following the discussion in \cite{dM05a} and \cite{HoTs10a}, a marginally trapped surface is one
whose expansion along the incoming and outgoing orthogonal, null geodesics
is non-positive.  On the boundary $\Sigma$, the expansion scalars are
given by 
\begin{align}\label{eq2:26jun13}
\hat{\theta}_{\pm}= \mp (n-1)\hat{H}+\text{tr}_{\hat{g}}\hat{K}-\hat{K}(\hat{\nu},\hat{\nu}),
\end{align}
where $(n-1)\hat{H} = \text{div}_{\hat{g}}\hat{\nu}$ is the mean extrinsic curvature of $\Sigma$ and $\hat{\nu}$ is
the outward pointing, unit normal vector field to $\cM$.  Therefore,
the surface $\Sigma$ is called a marginally trapped surface if 
$\hat{\theta}_{\pm} \leqs 0$.  See \cite{SD04,dM05a,Wald84} for details.

The problem we are interested in is to obtain solutions to the Einstein constraints for which $\theta_{\pm} \leqs 0$.  In order to
formulate this problem as a determined system, we use the conformal method of Lichnerowicz,
Choquet-Bruhat and York and the boundary conditions developed in \cite{HoTs10a}.
Using the conformal method, one can transform \eqref{eq1:27nov13}-\eqref{eq2:27nov13} into
a determined elliptic system by freely specifying conformal data, which consists of
a Riemannian manifold $(\cM,g)$, a transverse traceless tensor $\sigma$, 
a mean curvature function $\tau$, a non-negative energy density function $\rho$, and
a vector field $J$.  The Einstein constraints then become
\begin{align}
-\Delta \phi + c_n R \phi &+ b_n \tau^2\phi^{N-1}-c_n|\sigma + \cL W|^2\phi^{-N-1} - c_n \rho\phi^{-\frac{N}{2}} = 0, \label{eq3:27nov13} \\
&\Delta_{\mathbb{L}} W  + \frac{n-1}{n} \nabla \tau \phi^N+ J = 0 \label{eq4:27nov13},
\end{align}
where $\phi$ is an undetermined positive scalar and $W$ is an undetermined vector field.  In the above equations,
$R$ is the scalar curvature of $g$, $\cL$ is the conformal Killing operator
defined by 
$$
(\cL W)_{ij} = D_iW_j + D_jW_i - \frac2n \nabla^k W_k g_{ij},
$$
$\nabla$ and $\Delta$ are the connection and Laplacian associated with $g$, and $\Delta_{\mathbb{L}} = -\div\circ \cL$
is the vector Laplacian.  The constants $N$, $c_n$ and $b_n$ are dimensional constants given by
$$
N  = \frac{2n}{n-2}, \quad c_n = \frac{n-2}{4(n-1)}, \quad b_n = \frac{n-2}{4n}.
$$

Combining the the conformal method with the boundary conditions on $\theta_{\pm}$ in \eqref{eq2:26jun13},
one obtains the boundary conditions given in \cite{HoTs10a}.  
In particular,
we will be interested in the case when $\hat{\theta}_- = \theta_- \leqs 0$ is freely specified.  In this
case, the boundary conditions in \cite{HoTs10a} are
\begin{align}
\partial_{\nu}\phi+d_nH\phi+\left(d_n \tau - \frac{d_n}{n-1}\theta_- \right)&\phi^{\frac{N}{2}}-\frac{d_n}{n-1}S(\nu,\nu)\phi^{-\frac{N}{2}} = 0 \quad \text{on $\Sigma$}, \label{eq5:11july13}\\
(\cL \biw)(\nu,\cdot) &= \bV\quad \text{on $\Sigma$}, \label{eq4:8aug13}\\
S(\nu,\nu) =  \bV(\nu) +\sigma(\nu,\nu) &= ((n-1)\tau + |\theta_-|)\psi^N \geqs 0 \quad \text{on $\Sigma$}. \label{eq2:20dec13}
\end{align}
In~\eqref{eq5:11july13}, $H$ is the rescaled
extrinsic curvature for the boundary, $\nu = \phi^{\frac{N}{2}-1}\hat{\nu}$ is the rescaled normal vector field, and
$d_n = \frac{n-2}{2}$ is a dimension dependent constant.  The operators $\partial_{\nu}$ and $\cL$ are defined with respect to the specified metric $g$.  In order to guarantee that $\theta_+ \leqs 0$, the scalar function $\psi$
is chosen so that $\phi \leqs \psi$.  In general, we are interested in solving the coupled conformal system \eqref{eq3:27nov13}-\eqref{eq4:27nov13} with the boundary conditions 
\eqref{eq5:11july13}-\eqref{eq2:20dec13}.  We will refer to the boundary conditions \eqref{eq5:11july13}-\eqref{eq2:20dec13} 
with the added condition that $\phi \leqs \psi$ on $\cM$ 
as  {\bf marginally trapped surface boundary conditions}, or more simply as {\bf marginally trapped surface conditions}.

Our problem can now be expressed as a nonlinear, elliptic system of equations with Robin boundary conditions
that is of the form
\begin{equation}\label{eq1:19dec13}
\begin{aligned}
-\Delta \phi + c_n R \phi + b_n \tau^2\phi^{N-1}-c_n|\sigma + \cL W|^2\phi^{-N-1} - &c_n \rho\phi^{-\frac{N}{2}} = 0 ~~~~\text{on $\cM$}, \\
\partial_{\nu}\phi+d_nH\phi+\left(d_n \tau - \frac{d_n}{n-1}\theta_- \right)\phi^{\frac{N}{2}}-\frac{d_n}{n-1}S(&\nu,\nu)\phi^{-\frac{N}{2}} = 0 ~~~~ \text{on $\Sigma$}, \\
\Delta_{\mathbb{L}} W  + \frac{n-1}{n} \nabla \tau \phi^N+ J &= {\bf 0}  ~~~~~\text{on $\cM$},\\
(\cL \biw)(\nu,\cdot) &= \bV~~~~\text{on $\Sigma$}.
\end{aligned}
\end{equation}
One solves \eqref{eq1:19dec13} for $(\phi,W)$ and then constructs a solution to the constraints from
\begin{equation}\label{eq5:27nov13}
\begin{aligned}
&\hat{g}_{ab} =  \phi^{\frac{4}{n-2}}g_{ab}, \quad \hat{\rho} = \phi^{-\frac32 N +1}\rho, \quad \hat{J} = \phi^{-N}J,\\
&\hat{K}^{ab} = \phi^{2\frac{(n+2)}{(n-2)}}(\sigma + \cL W)^{ab} + \frac{\tau}{n} \phi^{\frac{4}{n-2}} g^{ab}. 
\end{aligned}
\end{equation}
If $\phi \leqs \psi$, the expansion scalars associated with $(\hat{g},\hat{K})$ will satisfy $\theta_{\pm} \leqs 0$.
In this case, $(\hat{g},\hat{K})$ will be a solution to the coupled system which satisfies the marginally trapped
surface conditions.

Boundary value problems similar to \eqref{eq1:19dec13} were first studied in the constant mean curvature or CMC case.
In \cite{SD04} and \cite{dM05a}, Dain and Maxwell proved the existence of apparent horizon
solutions in this setting, with slight variations on the boundary condition \eqref{eq5:11july13}.  Then in \cite{HoTs10a},
Holst and Tsogtgerel assembled a general collection of boundary conditions leading to marginally trapped surfaces that
 included the conditions of Maxwell and Dain, and then proved the existence of solutions to the Lichnerowicz problem on compact manifolds with 
 boundary with simplifications of these condtitions.  
 It is important to note that the conditions in \cite{HoTs10a} imply an additional coupling between $W$ and $\phi$ on the boundary, so even
in the constant mean curvature case, the equations do not decouple. 
Holst and Tsogtgerel intentionally ignored this coupling in order to develop results for the Lichnerowicz equation alone as the first step in a program for the coupled system, and therefore did not construct solutions to the constraints satisfying the marginally trapped surface conditions.
Their work then provided the mathematical framework for \cite{CHT13}, where Holst, Meier, and Tsogtgerel showed that non-CMC solutions to the constraints exist satisfying the marginally trapped surface boundary conditions.

{\bf\em Outline of the Paper.}
The remainder of the paper is organized as follows.
In Section~\ref{sec:notation}, we introduce some basic notation and terminology in order to allow us to give a fairly complete overview of the main results in Section~\ref{sec:main}.
We then develop some further notation and some basic supporting results in Section~\ref{sec:sobolev}.
The criticial barrier (sub- and supersolution) constructions needed for our main results are then given in Section~\ref{sec:barriers}.
The Schauder-based fixed-point framework is outlined in Section~\ref{sec:proof}, followed by a proof of our main far-from-CMC result.
A separate near-CMC result is then given in Section~\ref{sec:implicit}, based on the Implicit Function Theorem rather than a fixed-point argument.
Some supporting results we need that supplement existing literature on this problem are given in Appendix~\ref{sec:app}.

\section{Asymptotically Euclidean Manifolds and Harmonic Functions}
\label{sec:notation}

In this section, we introduce some basic notation and terminology in order to
give an overview of the main results in Section~\ref{sec:main}.
We will develop some further notation and some basic supporting
results in Section~\ref{sec:sobolev} before giving the proofs of the
main results in Sections~\ref{sec:barriers}--\ref{sec:implicit}.

{\bf\em Asymptotically Euclidean Manifolds.}
An $n$-dimensional, asymptotically Euclidean manifold $(\cM,g)$ is a non-compact Riemannian manifold, 
possibly containing a boundary, that can be decomposed into a compact
set $K$ and a finite number of ends $E_1,\cdots, E_k$. Each $E_j$ is diffeomorphic
to the exterior of a ball in $\mathbb{R}^n$, and on each end
the metric $g$ tends towards the Euclidean metric $g_E$.

To formalize this definition, we recall the definition of the weighted Sobolev space $W^{k,p}_{\delta}(\cM)$ of scalar functions.
(See \cite{Ba86} for an in depth discussion.)
For $k \in \mathbb{N}$, $p \geqs 1$, a given function $u \in W^{k,p}_{\delta}$ if
\begin{align}\label{aeweight}
\|u\|_{W^{k,p}_{\delta}} = \sum_{|\beta| \leqs k} \|r^{\delta -\frac{n}{p}+|\beta|}\partial^{\beta} u \|_{L^p} < \infty.
\end{align}
In the above norm, partial derivatives are taken with respective to a fixed coordinate chart and
$r$ is a smooth positive function that agrees with $|x|$ on each end $E_j$. 
For example, we may take $r(x) = \sqrt{1+D(x,p_0)^2}$, where $D(x,p_0)$ denotes the distance
from $x$ to an arbitrary fixed point $p_0 \in K$.  We will also consider the space of weighted, continuous
functions $C^k_{\delta}(\cM)$, whose norm is
given by
$$
\|u\|_{C^k_{\delta}} = \sum_{|\alpha|\leqs k}\sup_{x\in \cM}(r^{-\delta+|\alpha|}|\partial^{\alpha}u|).
$$
The weighted Sobolev spaces and continuous spaces are related by the continuous embedding
$W^{k,p}_{\delta}\hookrightarrow C^{0}_{\delta}$, which holds if $k > n/p$.  
 
If $T^{a_1a_2,\cdots a_r}_{b_1b_2,\cdots b_s}$ is an $(r,s)$-tensor, we may define the point
value of $T$ by
$$
|T| = (T^{a_1a_2,\cdots a_r}_{b_1b_2,\cdots b_s}T_{a_1a_2,\cdots a_r}^{b_1b_2,\cdots b_s})^{\frac{1}{2}}.
$$
The above norms can then be applied to $|T|$, which allows one to
consider weighted spaces $W^{k,p}_{\delta}(T^r_s\cM)$ of $(r,s)$-tensors.  
In particular, we let $\bW^{k,p}_{\delta} = W^{k,p}_{\delta}(T\cM)$ denote the weighted Sobolev space of
vector fields on $\cM$.

We say that {\em $g$ tends towards $g_{E}$ and is $W^{k,p}_{\delta}$-asymptotically Euclidean} if, for some $\delta < 0$, 
\begin{align}\label{AEmetric}
g-g_{E} \in W^{k,p}_{\delta}.
\end{align}
We note that if $g - A_i^{N-2}g_E \in W^{k,p}_{\delta}(E_i)$ for each
$E_i$, $g$ is also $W^{k,p}_{\delta}$-asymptotically Euclidean
given that it will satisfy \eqref{AEmetric} with an appropriate change
of coordinates (cf \cite{DIMM13}). 
Using these weighted spaces, we define an asymptotically
Euclidean data set.  As in \cite{DIMM13}, we say the data set $(\cM,g,K,\rho,J)$ is asymptotically Euclidean if
for some $\delta < 0$, $g-g_{E} \in W^{k,p}_{\delta}$, $K \in W^{k-1,p}_{\delta-1}$, and
$\rho, J \in W^{k-2,p}_{\delta-2}$.

In the event that $(\cM,g)$ has a boundary $\Sigma$, we
consider the Sobolev spaces $W^{k,p}(\Sigma)$ for $k\in \mathbb{N}$ and $p>1$.
These Banach spaces consist of the set of all
functions $u$ such that
\begin{align}
\|u\|_{k,p;\Sigma} = \sum_{l \le k}\|\nabla^l u\|_{p;\Sigma} < \infty,
\end{align} 
where the connection $\nabla$ and integration are
with respect to the boundary metric induced by $g$.  This
definition can be extended to obtain the fractional order Sobolev spaces
$W^{s,p}(\Sigma)$ with $s \in \mathbb{R}$.
See \cite{HNT07b} for more details, including general results concerning multiplication properties of these spaces.

{\bf\em Asymptotic Limits and Harmonic Functions.}
We will seek solutions $(\phi,W)$ to the conformal equations where $\phi$ has fairly general asymptotic behavior.
The following framework for representing this behavior is a generalization of the approach developed in \cite{DIMM13}, suitable for our needs here.

Given constants $A_1,\cdots, A_k$, we seek solutions
such that $\phi \to A_i$ on each end $E_i$.  Let 
$\mathcal{H}$ denote the space of smooth, harmonic
functions with zero Neumann boundary conditions on $\Sigma$.  By Proposition~\ref{prop1:23dec13}, there exists a unique
$\omega\in \mathcal{H}$ such that $\omega \to A_i$ on $E_i$.  Therefore $\mathcal{H} \cong \mathbb{R}^k$ and  
if $\gamma < 0$, $\phi-\omega \in W^{2,p}_{\gamma}$ implies that $\phi \to A_i$
on each end.  So $\omega$ encodes the asymptotic behavior of $\phi$.

Because we can represent the asymptotic behavior of our solution $\phi$
by an element in $\omega \in \mathcal{H}$, we will seek solutions of the form
$\phi = \omega + u$, where $u \in W^{2,p}_{\gamma}$.
Therefore we define the space
$$
\mathcal{H}+W^{k,p}_{\delta} =\{\omega + u~|~\omega \in \mathcal{H},~u \in W^{k,p}_{\delta}\}.
$$
If $C^0$ denotes the space of continuous functions on $\cM$, 
we note that if $k > n/p$ there exists a compact embedding 
\begin{align}\label{embed}
\mathcal{H}+W^{k,p}_{\delta} \hookrightarrow C^0,
\end{align}
given that $\mathbb{R}^k\oplus W^{k,p}_{\delta} \hookrightarrow \mathbb{R}^k\oplus C^0$ compactly
and $\mathcal{H}+C^0 \subset C^0$.  

We will also need a way to compare the asymptotic limits of two functions $f,g$.
We say $f$ is {\bf asymptotically bounded below} by $g$ if 
$$\lim_{|x|\to \infty}f \geqs \lim_{|x|\to \infty}g \quad \text{ on each $E_i$},$$
and $f$ is {\bf asymptotically bounded above} by $g$ if $g$ is asymptotically bounded below by $f$.  
Finally, given $g\leqs h$ we say that $f$ is {\bf asymptotically
bounded} by $g$ and $h$ if $f$ is asymptotically bounded below by $g$ and asymptotically bounded above by $h$. 

{\bf\em Yamabe Invariant on Asymptotically Euclidean Manifolds with Boundary.}
To finish the discussion of notation needed for stating our main results, let us recall the definition of the Yamabe invariant on asymptotically Euclidean manifolds $\cM$ with boundary $\Sigma$.
Define the following functional for compactly supported functions $f\in C^{\infty}_c$: 
\begin{align}\label{eq1:2dec13}
Q_g(f) = \frac{ \int_{\cM} |\nabla f|^2 + c_n Rf^2~dV + \int_{\Sigma} d_n Hf^2~dA}{\|f\|^2_{L^{\frac{2n}{n-2}}}}.
\end{align}
Then as in \cite{dM05a}, the Yamabe invariant on $\cM$ is
\begin{align}\label{eq2:2dec13}
\cY_g = \inf_{f\in C^{\infty}_c(\cM),~~f\ne 0} Q_g (f).
\end{align}

\section{Overview of the Main Results}
\label{sec:main}

The main results for this paper concern the existence of far-from-CMC and near-CMC solutions to the conformal formulation 
of the Einstein constraint equations on an asymptotically Euclidean, $n$-dimensional manifold $\cM$
with compact boundary $\Sigma$.
We assume that the boundary consists of $m$ distinct components
\begin{equation}\label{eq6:11july13}
\Sigma =  \cup_{1}^m \Sigma_i , \quad \Sigma_i \cap \Sigma_j = \emptyset.
\end{equation}
Here, each component $\Sigma_i$ represents a marginally trapped surface, and
$\cM$ is an embedded submanifold of some manifold $\cN$. 
We view $\cM$ as the result of excising trapped regions $C_i$ with boundary $\Sigma_i$ from $\mathcal{N}$.
Therefore, the following theorems provide conditions under which we may
obtain solutions to the Einstein constraints outside of the singular trapped regions $C_i$ with
minimal assumptions on the mean curvature $\tau$.

Our first Theorem is a far-from-CMC result in that it places no restrictions on the mean curvature function $\tau$.
However, to compensate for this assumption we require smallness assumptions on the other data.

\begin{theorem}\label{thm1:2dec13}{\bf (Far-From-CMC)}
Suppose that $(\cM,g)$ is asymptotically Euclidean of class $W^{2,p}_{\gamma}$ with $p>n$
and $2-n < \gamma < 0$.  Assume that $2-n < \delta< \gamma/2$, and the
data satisfies:
\begin{itemize}
\item[$\bullet$] $ g\in \cY^+$,
\item[$\bullet$] $\tau \in W^{1,p}_{\delta-1}$,
\item[$\bullet$] $ \sigma \in  W_{\delta-1}^{1,2p}$ with $\|\sigma\|_{L^{\infty}_{\delta-1}}$ sufficiently small,
\item[$\bullet$] $ \rho \in L^{\infty}_{\gamma-2}$ with $\|\rho\|_{L^{\infty}_{\delta-2}}$ sufficiently small,
\item[$\bullet$] $ J \in {\bf L}^{p}_{\delta-2}$ with $\|J\|_{L^{p}_{\delta-2}}$ sufficiently small,
\item[$\bullet$] $ \theta_- \in W^{1-\frac1p,p}(\Sigma)$, ~~~ $\theta_- < 0$,
\item[$\bullet$] $ \bV \in \bW^{1,p}$, ~~~$\bV|_{\Sigma} =\left(((n-1)\tau + |\theta_-|/2)\psi^N-\sigma(\nu,\nu)\right) \nu$,
\item[$\bullet$] $ ((n-1)\tau + |\theta_-|/2) > 0$ and $\|(n-1)\tau + |\theta_-|/2\|_{W^{1-\frac1p,p}(\Sigma)}$ sufficiently small.
\end{itemize}
Then on each end $E_i$ there exists an interval $\mathcal{I}_i\subset (0,\infty)$ such that if $A_i \in \mathcal{I}_i$ are
freely specified constants and $\omega$ is the associated harmonic function, 
there exists a solution $(\phi,W)$ to the conformal equations with boundary conditions \eqref{eq5:11july13}-\eqref{eq2:20dec13}
such that $\phi -\omega \in W^{2,p}_{\gamma}$ and $W \in W^{2,p}_{\delta}$.  Moreover, the function $\psi$ can be chosen 
so that $(\phi,W)$ satisfies the marginally trapped surface boundary conditions.
\end{theorem}
\begin{proof}
The proof is given in Section~\ref{sec:proof}.
\end{proof}

The following Theorem complements Theorem~\ref{thm1:2dec13} by showing
that smallness assumptions on $\tau$ replace the need for smallness assumptions
on $\sigma$ and $\rho$.  Given that the proof relies on the Implicit Function Theorem,
solutions will be unique in this case.

\begin{theorem}\label{thm2:2dec13}{\bf  (Near-CMC with $g\in \cY^+$)}
Suppose that $(\cM,g)$ is asymptotically Euclidean of class $W^{2,p}_{\gamma}$ with $p>n$
and $2-n < \gamma < 0$.  Assume that $2-n < \delta< \gamma/2$, and the
data satisfies:
\begin{itemize}
\item[$\bullet$] $g\in \cY^+$,
\item[$\bullet$] $\|\tau\|_{W^{1,p}_{\delta-1}}$ is sufficiently small, and~~~$\tau \geqs 0$ on $\Sigma$,
\item[$\bullet$] $ \sigma \in  W_{\gamma-1}^{1,2p},$
\item[$\bullet$] $ \rho \in L^{p}_{\gamma-2},$
\item[$\bullet$] $ \|J\|_{{\bf L}^{p}_{\delta-2}}$ is sufficiently small,
\item[$\bullet$] $ \theta_- = 0$,
\item[$\bullet$] $ \bV \in \bW^{1,p}$, ~~~$\bV|_{\Sigma} =\left(((n-1)\tau)\phi^N-\sigma(\nu,\nu)\right) \nu$.
\end{itemize}
Then if $A_i \in (0,\infty)$ are freely specified constants on each end $E_i$ and $\omega$ is the associated harmonic
function, there exists 
a unique solution $(\phi,W)$ to the conformal equations with marginally trapped surface boundary conditions
such that $\phi-\omega \in W^{2,p}_{\gamma}$ and $W \in W^{2,p}_{\delta}$. 
\end{theorem}
\begin{proof}
The proof follows from Corollary~\ref{cor1:19feb14} given in Section~\ref{sec:implicit}.
\end{proof}

Our final Theorem states that we may replace the assumption that $g\in \cY^+$ with the
assumption that $R$ and $H$ are bounded from below in terms of $\tau $ and $|\theta_-|$.
In this case smallness assumptions are imposed on $\|\nabla\tau\|_{L^p_{\delta-1}}$ and 
$(2(n-1)\tau+|\theta_-|)$ on $\Sigma$.

\begin{theorem}\label{thm1:12feb14}{\bf  (Near-CMC with bounded $R$ and $H$)}
Suppose that $(\cM,g)$ is asymptotically Euclidean of class $W^{2,p}_{\gamma}$ with $p>n$
and $2-n < \gamma < 0$.  Assume that $2-n < \delta< \gamma/2$, and the
data satisfies:
\begin{itemize}
\item[$\bullet$] $\|\nabla\tau\|_{L^{p}_{\delta-2}}$ is sufficiently small,
\item[$\bullet$] $ \sigma \in  W_{\gamma-1}^{1,2p},$
\item[$\bullet$] $ \rho \in L^{p}_{\gamma-2},$
\item[$\bullet$] $ J \in { \bf L}^{p}_{\delta-2},$
\item[$\bullet$] $ \theta_- \in W^{1-\frac1p,p}(\Sigma)$, ~~~ $\theta_- < 0$,
\item[$\bullet$] $ \bV \in \bW^{1,p}$, ~~~$\bV|_{\Sigma} =\left(((n-1)\tau + |\theta_-|/2)\psi^N-\sigma(\nu,\nu)\right) \nu$,
\item[$\bullet$] $ (2(n-1)\tau + |\theta_-|) > 0$  is sufficiently small on $\Sigma$.
\end{itemize}
Let $A_i \in [1,\infty)$ be freely specified constants on each end $E_i$ and let $\omega$ be the associated harmonic
function.  Then if 
\begin{itemize}
\item $-c_nR \leqs b_n\tau^2$ on $\{x\in\cM:R(x) < 0\}$,

\item $-H \leqs (\tau+|\theta_-|/(n-1))$ on $\{x\in \Sigma:H(x) < 0\}$,

\end{itemize}
there exists a solution $(\phi,W)$ to the conformal equations with boundary conditions \eqref{eq5:11july13}-\eqref{eq2:20dec13}
such that $\phi -\omega \in W^{2,p}_{\gamma}$ and $W \in W^{2,p}_{\delta}$.  Moreover, the function $\psi$ can be chosen 
so that $(\phi,W)$ satisfies the marginally trapped surface boundary conditions.
\end{theorem}
\begin{proof}
The proof is given in Section~\ref{sec:proof}.
\end{proof}

\begin{remark}
The conditions that $-c_nR \leqs b_n\tau^2$ on $\{x\in\cM:R(x) < 0\}$ and ${\|\tau+|\theta_-|/(n-1)\|_{1-\frac1p,p;\Sigma}}$ and ${\|\nabla\tau\|_{L^p_{\delta-1}}}$ be sufficiently small place restrictions on the metric $g$.  
Namely, this method might not be applicable for
metrics $g$ which have large, negative scalar curvature.  Similarly, the condition that
$-H \leqs (\tau+|\theta_-|/(n-1))$ on $\{x\in \Sigma:H(x) < 0\}$ and $\tau+|\theta_-|/(n-1)$ be small
on $\Sigma$ imposes conditions on the boundary.  It is possible that these boundedness conditions
on $R$ and $H$ relate to the positive Yamabe condition, however this relationship is not well understood (cf. \cite{yCBjIjY00}).  
\end{remark}

\begin{remark}
In Theorems~\ref{thm1:2dec13}-\ref{thm1:12feb14} we assume the existence of a vector field $\bV \in \bW^{1,p}$ which satisfies
$$V|_{\Sigma} =\left(((n-1)\tau + |\theta_-|/2)\psi^N-\sigma(\nu,\nu)\right) \nu.$$  In the proof
of Proposition~\ref{prop1:20dec13} we explicitly construct a vector field satisfying these assumptions.
\end{remark}

Theorem \ref{thm2:2dec13} follows from a variation of the Implicit Function Theorem
argument developed in \cite{yCBjIjY00}, where one perturbs $\tau$ and $J$ from zero
to obtain a small neighborhood of solutions about a known solution to the decoupled conformal equations.  
The proofs of Theorems~\ref{thm1:2dec13} and \ref{thm1:12feb14} follow from a variation of the
Schauder fixed point argument developed in \cite{HNT07a,HNT07b} for compact manifolds.  This approach
was adapted to asymptotically Euclidean manifolds in \cite{DIMM13}, and we use a variation of
that argument.  We briefly outline this method of proof below.

For a fixed $W$, we let $\mathcal{N}(\phi,W)$ denote the Lichnerowicz operator on the left of \eqref{eq3:27nov13} with boundary
operator on the left of \eqref{eq5:11july13}.  In this notation, a solution to the coupled system \eqref{eq1:19dec13} satisfies
$\mathcal{N}(\phi,\mathcal{S}(\phi)) = 0$, where $\mathcal{S}(\phi)$ denotes the solution to the momentum constraint \eqref{eq4:27nov13} with boundary conditions \eqref{eq4:8aug13}
for a given $\phi$.

A rough outline of the Schauder fixed point argument is as follows.
If $C_+^0$ denotes the spaces of positive, continuous functions,
for a given $\phi \in \mathcal{C}^0_+$ and $\psi \in W^{2,p}(\Sigma)$ let
$W=\mathcal{S}(\phi)$ denote the momentum constraint solution map
with boundary condition \eqref{eq4:8aug13}-\eqref{eq2:20dec13}.
Similarly, for a given $W \in W^{2,p}_{\delta}$ and sub-and supersolutions $\phi_- \leqs \phi_+$, 
Theorems~\ref{thm1:23jan14} and \ref{uniq} in Appendix~\ref{sec:app} imply that for a given
$\omega \in \cH$ which is asymptotically bounded by $\phi_-$ and $\phi_+$, there exists
a unique solution to $\mathcal{N}(\phi,W) = 0$ such that $\phi - \omega \in W^{2,p}_{\gamma}$. 
Therefore we let $\phi = T(W)$ denote the Hamiltonian constraint solution map
with boundary conditions \eqref{eq5:11july13}.  
If $\it{i}$ is the compact
embedding $\mathcal{H}\oplus W^{2,p}_{\gamma}\hookrightarrow C^0$ defined
in \eqref{embed},
then we set 
\begin{align}\label{fixpt}
\mathcal{G}(\phi) = \it{i}(T(\mathcal{S}(\phi))).
\end{align}
A solution to the coupled system with the specified boundary conditions will be a fixed point of this map.  In order to apply the Schauder fixed point
argument in \cite{HNT07b}, we must show that this map is compact and invariant on a certain subset of $C^0_+$.

The primary difficulty in applying this fixed point argument is in constructing the closed, bounded, and convex subset
of $C^0$ on which the map $\mathcal{G}(\phi)$ is invariant.  The construction of this set requires global
sub- and supersolutions $\phi_-$ and $\phi_+$ of the Hamiltonian constraint, and once these are obtained
the process is fairly straightforward.  See \cite{HNT07b,CHT13,DIMM13}.  In Section~\ref{sec:barriers}, we
will construct global sub-and super-solutions for the Hamiltonian constraint and in Section~\ref{sec:proof} we use this
framework to prove Theorems~\ref{thm1:2dec13} and \ref{thm1:12feb14}.  Then in Section~\ref{sec:implicit} we use 
the Implicit Function Theorem to obtain the near-CMC results in Theorem~\ref{thm2:2dec13}.

\section{Properties of Linear Operators on Weighted Sobolev Spaces}
\label{sec:sobolev}

Using the definition of Yamabe invariant given in~\ref{eq1:2dec13}, we compile some useful facts from \cite{dM05a} about the operators
$$\cP_1 = (-\Delta + c_n R, \partial_{\nu} + d_nH) \quad \text{ and} \quad \cP_2 = (-\Delta_{\mathbb{L}}, B),$$
where $BW = \cL W(\nu,\cdot)$. 
In the following proposition, we summarize the properties of both $\cP_1$ and $\cP_2$.  We write
$L^p_{\delta-2}\times W^{1-\frac1p,p}$ to indicate $L^p_{\delta-2}(\cM)\times W^{1-\frac1p,p}(\Sigma)$
in the case of $\cP_1$ and $L^p_{\delta-2}(T\cM)\times W^{1-\frac1p,p}(T\Sigma)$ in the case of $\cP_2$.

\begin{proposition}\label{prop1:14oct13}
Suppose $(\cM,g)$ is asymptotically Euclidean of class $W^{k,p}_{\gamma}$ with $k \geqs 2$, $k > n/p$, and $2-n < \delta < 0$.
Then $\cP_i: W^{2,p}_{\delta} \to L^p_{\delta-2}\times W^{1-\frac1p,p}$ is Fredholm with index zero.  Moreover, if $\cY_g > 0$, then $\cP_1$ is an 
isomorphism and if $p>n$ or $\cM$ possesses no conformal Killing fields, then $\cP_2$ is an isomorphism.  Finally, if $\cP_i$ is
an isomorphism and $\cP_i v = (f,g) \in L^p_{\delta-2}\times W^{1-\frac1p,p}$, then there exists $C  > 0$ such that the following estimate is satisfied:
\begin{align}\label{eq3:2dec13}
\|v\|_{W^{2,p}_{\delta}} \leqs C\left( \|f \|_{L^p_{\delta-2}} + \| g \|_{W^{1-\frac1p,p}}\right).
\end{align}
\end{proposition}
\begin{proof}
See Proposition 1, Proposition 3, Proposition 6 and Theorem 3 in \cite{dM05a}.
\end{proof}

With Proposition~\ref{prop1:14oct13} in hand, we can now prove the following important estimate in the case when $k=2$.
This result is based on a similar estimate in \cite{Dilt13b}.

\begin{proposition}\label{prop1:17oct13}
Suppose $(\cM,g)$ is asymptotically Euclidean of class
$W^{2,p}_{\gamma}$ with $n < p $,
and let $r$ be the function defined in
\eqref{aeweight}.
Then for a given $\phi \in L_+^{\infty}$, if $W \in W^{2,p}_{\delta}$ is
solves the momentum constraint \eqref{eq4:27nov13} with boundary conditions
\eqref{eq4:8aug13}, where
$2-n \leqs \delta<0$, there exists
$C>0$ such that the following estimate holds:
\begin{align}\label{eq2:14oct13}
\|\cL W\|_{\infty} \leqs Cr^{\delta-1}\left(\|\nabla \tau\|_{L^p_{\delta-2}}\|\phi\|^N_{\infty}+\|J\|_{L^p_{\delta-2}}+ \|\bV\|_{W^{1-\frac1p,p}(T\Sigma)} \right)
\end{align}
\end{proposition}
\begin{proof}
By Proposition~\ref{prop1:14oct13} we have
\begin{align}\label{eq4:14oct13}
\|W\|_{W^{2,p}_{\delta}} \leqs c\left(\|\nabla \tau\phi^N\|_{L^p_{\delta-2}}+ \|J\|_{L^{p}_{\delta-2}} + \|\bV\|_{W^{1-\frac1p,p}(T\Sigma)}\right).
\end{align}
The continuous embedding $W^{1,p}_{\delta-1} \hookrightarrow C^{0}_{\delta-1}$ implies that
$$
\|\cL W\|_{C^0_{\delta-1}} \leqs C_1\|LW\|_{W^{1,p}_{\delta-1}}\leqs C_2\|W\|_{W^{2,p}_{\delta}},
$$
and combining this with estimate \eqref{eq4:14oct13} we have
\begin{align}
 \|\cL W\|_{C^0_{\delta-1}}\leqs C\left(\|\phi \|^N_{\infty}\|\nabla \tau\|_{L^{q}_{\delta-2}}+ \|J\|_{L^{p}_{\delta-2}}+ \|\bV\|_{W^{1-\frac1p,p}(T\Sigma)} \right).
\end{align}
The above estimate and the definition of the $C^0_{\delta-1}$ norm imply the result.
\end{proof}

Propositions~\ref{prop1:14oct13} and \ref{prop1:17oct13} will be essential in determining our global barriers.  
In particular, Proposition~\ref{prop1:17oct13} is our primary tool to control the point-wise values of the solution of the
momentum constraint $W= \mathcal{S}(\phi)$ in terms of $\phi$.  This will be vital when we construct our global supersolution
in the next section.

\section{Barriers for the Hamiltonian constraint}
\label{sec:barriers}

A critical component of fixed-point arguments for nonlinear elliptic equations are the development of {\em a priori} estimates, and/or sub- and supersolutions.
These so-called barriers are an essential component for building the $\mathcal{G}$-invariant set necessary for our fixed point argument, where we
recall that $\mathcal{G}$ is the nonlinear fixed point operator defined in \eqref{fixpt}.
Therefore, in this section we will develop several global sub-and supersolution constructions for the Hamiltonian constraint equation \eqref{eq3:27nov13}. 

If $\gamma$ is the trace operator associated with $\Sigma$, we define the 
operators
\begin{align}
A_L(\phi) = &\left( \begin{array}{c} -\Delta \phi +c_na_R \phi \\ \gamma (\partial_{\nu}\phi) + d_nH(\gamma\phi)  \end{array}\right)\nonumber,\\
F(\phi,W) = &\left( \begin{array}{c} b_n \tau^2\phi^{N-1}-c_n |\sigma + \cL W|^2 \phi^{-N-1}-c_n\rho\phi^{-\frac{N}{2}} \\ \left(d_n\gamma\tau- \frac{d_n}{n-1}\theta_-\right)(\gamma(\phi))^{\frac{N}{2}}-\frac{d_n}{n-1}S(\nu,\nu)(\gamma(\phi))^{-\frac{N}{2}}\end{array}\right). \nonumber
\end{align}
The Hamiltonian constraint with boundary conditions
\eqref{eq5:11july13}-\eqref{eq2:20dec13} can be written succinctly as 
\begin{eqnarray}
\cN(\phi,W) = A_L(\phi) + F(\phi,W) & = & 0. 
\label{E:ham-abstract} 
\end{eqnarray}
Using this notation, we recall that for a given vector field $W$, if the functions $\phi_-$ and $\phi_+$
satisfy
$$
\cN(\phi_-,W) \leqs 0 \quad \text{and} \quad \cN(\phi_+,W) \geqs 0,
$$
then $\phi_-$ is called a subsolution and $\phi_+$ is a supersolution.

As in \cite{HNT07b}, to obtain a fixed point of the coupled conformal equations we
require a slightly more restrictive class of sub- and supersolutions.  If $W=W(\phi)$ denotes the solution of the
momentum constraint for a given $\phi$ and
$$
\cN(\phi_-,W(\phi)) \leqs 0\quad \text{for all }\quad \phi \geqs \phi_-,
$$
then $\phi_-$ is a {\bf global subsolution}.  Similarly, $\phi_+$ is a { \bf global supersolution} if
$$
\cN(\phi_+,W(\phi)) \geqs 0 \quad \text{for all }\quad \phi \leqs \phi_+.
$$

In the following discussion we will require that when the vector field $ W \in \biW^{2,p}_{\delta}$ is given by the solution of the momentum constraint equation \eqref{eq4:27nov13} with the source term $\phi\in L^{\infty}$,
\begin{equation}\label{CS-aLw-bound}
|\cL W|^2 \leqs 
r^{2\delta -2}(\ttk_1 \, \|\phi\|_{\infty}^{2N} + \ttk_2),
\end{equation}
with some positive constants $\ttk_1$ and $\ttk_2$.
The following proposition justifies this bound.
\begin{proposition}\label{prop1:5dec13}
Let the assumptions of Proposition~\ref{prop1:17oct13} hold with $p \in (n,\frac{2\alpha+1}{2})$, $\alpha > n$, and $s \in (1+\frac{(n-1)}{p}-\frac{(n-1)}{\alpha}, 1+\frac{(n-1)}{p})$. 
Suppose that $W$ satisfies the momentum constraint \eqref{eq4:27nov13} with boundary conditions
\eqref{eq4:8aug13}-\eqref{eq2:20dec13}, where $\bV(\nu) = ((n-1) \tau + |\theta_-|/2)\psi^N~~$.
Then $W$ satisfies the bound \eqref{CS-aLw-bound} with
\begin{align}\label{e:k1k2}
&\ttk_1 = 2 C_1^2 \|\nabla \tau\|_{L^p_{\delta-2}}^2,\\
&\ttk_2 = 2C_2^2\left(\|J\|_{L^p_{\delta-2}}+\|\sigma(\nu,\nu)\|_{1-\frac{1}{p},p;\Sigma} \right. \nonumber \\
&\qquad + \left. \|((n-1)\tau+|\theta_-|/2)\|_{1-\frac{1}{p},p;\Sigma}\|\psi\|^{N-1}_{\infty}\|\psi\|_{s,p;\Sigma}\right)^2. \nonumber
\end{align}
\end{proposition}
\begin{proof}
Proposition~\ref{prop1:17oct13} implies that
\begin{align}\label{eq1:10dec13}
& \|\cL W\|_{\infty}^2 \leqs C^2r^{2\delta-2}\left(\|\phi\|^{N}_{\infty}\|\nabla \tau\|_{L^p_{\delta-2}} +\|J\|_{L^p_{\delta-2}} + \|\bV\|_{1-\frac{1}{p},p;\Sigma}\right)^2. 
\end{align}
If $\bV(\nu) = ((n-1) \tau + |\theta_-|/2)\psi^N-\sigma(\nu,\nu)$ as in Theorem~\ref{thm1:2dec13}, then $ \bV|_{\Sigma} = \bX +{\bf Y}$, where
$\bX(\nu) = \bV(\nu) $ and ${\bf Y}(\nu) = {\bf 0}$.  In practice, we will assume that ${\bf Y} = {\bf 0}$ (cf. Proposition~\ref{prop1:20dec13}), and 
we have that
\begin{align}
\|\bV\|_{1-\frac{1}{p},p;\Sigma} \leqs C\left(\|((n-1)\tau+|\theta_-|/2)\psi^N\|_{1-\frac{1}{p},p;\Sigma} + \|\sigma(\nu,\nu)\|_{1-\frac{1}{p},p;\Sigma}  \right).
\end{align}
We now apply Lemma A.21 from \cite{HNT07b} with $\sigma = 1-\frac1p$, $p =q$, and $s \in (1+\frac{(n-1)}{p}-\frac{(n-1)}{\alpha}, 1+\frac{(n-1)}{p})$.
This gives us that
\begin{align*}
\|((n-1)\tau+|&\theta_-|/2)\psi^N\|_{1-\frac{1}{p},p;\Sigma} \\
& \leqs C\|((n-1)\tau+|\theta_-|/2)\|_{1-\frac{1}{p},p;\Sigma}\left( \|\psi^N\|_{\infty} + N\|\psi^{N-1}\|_{\infty}\|\psi\|_{s,p;\Sigma} \right).
\end{align*}
The embedding $W^{s,p}(\Sigma) \hookrightarrow L^{\infty}(\Sigma)$ implies that
$$
\|((n-1)\tau+|\theta_-|/2)\psi^N\|_{1-\frac{1}{p},p;\Sigma} \leqs C\|((n-1)\tau+|\theta_-|/2)\|_{1-\frac{1}{p},p;\Sigma}\|\psi\|^{N-1}_{\infty}\|\psi\|_{s,p;\Sigma},
$$
which combined with \eqref{eq1:10dec13} implies the result.
\end{proof}

In the following discussion, we let $\omega \in \cH$ be 
such that $\omega \to A_j>0$ on each end $E_j$.  
Additionally, given any scalar function $u\in L^{\infty}$,
we use the notation
\[
u^{\tiwedge}:= \mbox{ess~sup}\, u,\qquad
u^{\tivee}:= \mbox{ess~inf}\, u.
\]

We are now ready to construct our global barriers.  The following Theorem provides conditions under
which we can construct a global supersolution for $\mathcal{N}$ given that $W$ satisfies
the boundary conditions~\eqref{eq4:8aug13}-\eqref{eq2:20dec13}. For this
particular construction, if we want to freely specify $\psi \in W^{1-\frac1p,p}(\Sigma)$ we are required
to assume that $(2(n-1)\tau+|\theta_-|)$ is sufficiently small on $\Sigma$. 

\begin{theorem}{\bf (Far-From-CMC Global Supersolution )}\label{thm1:17oct13}
Suppose that $(\cM,g)$ is asymptotically Euclidean of class $W^{2,p}_{\gamma}$, with $n < p$
and $ \gamma \in (2-n, 0)$, and that $2-n <\delta <\gamma/2$.  Additionally assume that $\cY_g > 0$, $\tau \in W^{1,p}_{\delta-1}$, and that $\sigma \in L^{\infty}_{\delta-1}\cap W^{1,2p}_{\delta-1}$, $J \in L^p_{\delta-2}$ and 
$\rho \in L^{\infty}_{\delta-2}$ are sufficiently small.  Also assume that for $0<\psi \in W^{2,p}_{\delta}$, 
$(2(n-1)\tau+|\theta_-|)\psi^N > 0$ is sufficiently small on $\Sigma$.  Then there exists a global
super-solution $\phi_+>0$ to the Hamiltonian constraint with boundary condition \eqref{eq5:11july13}-\eqref{eq2:20dec13} such that $\phi_+ - \beta \omega \in W^{2,p}_{\gamma}$ for some $\beta > 0$ sufficiently small.
\end{theorem}
\begin{proof}
Let $\Lambda \in L^p_{\gamma-2}$ be a positive function that agrees with
$r^{\gamma -2}$ outside of a compact set and let $\lambda \in W^{1-\frac1p,p} $ be a positive
function  on $\Sigma$. Then by
Proposition~\ref{prop1:14oct13} and Proposition 3.2 in \cite{DIMM13} there exists solution $u \in W^{2,p}_{\gamma}$ solving
\begin{align}\label{eq12:17oct13}
-\Delta u +c_nR u = \Lambda - c_n \omega R,  \\
\partial_{\nu}u + d_nH u = \lambda -d_n\omega H. \nonumber
\end{align}
Let $\phi_+ = \beta(u+\omega)$, where $\beta>0$ will be determined.  By the
maximum principles~\ref{wmaxprinc} and \ref{smaxprinc} we have that $\phi_+ > 0$.
We recall from Proposition~\ref{prop1:5dec13} that we may bound $a_{\cL W}$ in terms
of the source function $\phi$.
Using this bound and the fact that $a_W = c_n|\sigma +\cL W|^2 \leqs 2|\sigma|^2 + 2|\cL W|^2$, we obtain the bound
\begin{align}\label{eq1:23jan14}
a^{\tiwedge}_{W} \leqs r^{2\delta -2}(K_1\|\phi\|^{2N}_{\infty}+K_2),
\end{align}
where
$$
K_1 = C_1\ttk_1 \quad \text{and} \quad K_2 = 2r^{2-2\delta}(\sigma^2)^{\tiwedge} + C_2\ttk_2,
$$
and $\ttk_1$ and $\ttk_2$ are the same constants in \eqref{e:k1k2}.  We let $W(\phi)$ denote a solution
to the momentum constraint for a given $\phi < \phi_+$ and define $\ttk_3 = (\frac{\sup \phi_+}{\inf \phi_+})^{2N}$.
Applying the Hamiltonian constraint \eqref{E:ham-abstract}
to $\phi_+$ and using the fact that $S(\nu,\nu) = ((n-1)\tau +|\theta_-|/2)\psi^N$, we obtain
\begin{align}
\cN(&\phi_+,W(\phi)) \nonumber \\
&= \left(\begin{array}{c}-\Delta \phi_+ +c_nR\phi_+ +b_n\tau^2\phi_+^{N-1} -a_{W}\phi_+^{-N-1}-c_n\rho\phi_+^{-N/2} \\ \partial_{\nu} \phi_+ + d_nH\phi_+ + (d_n\tau -\frac{d_n}{n-1}\theta_-)\phi_+^{\frac{N}2} - \frac{d_n}{n-1}(((n-1)\tau+|\theta_-|/2)\psi^N)\phi_+^{-\frac{N}2} \end{array}\right) \nonumber\\
                 &  \geqs \left( \begin{array}{c} \beta\Lambda  +\beta^{N-1} c_n\tau^2(u+\omega)^{N-1} -r^{2\delta-2}(K_1 (\phi^{\tiwedge})^{2N}+K_2)\phi_+^{-N-1}-c_n\rho\phi_+^{-N/2} \\ \beta \lambda + (d_n\tau -\frac{d_n}{n-1}\theta_-)\phi_+^{\frac{N}2} - \frac{d_n}{n-1}(((n-1)\tau+|\theta_-|/2)\psi^N)\phi_+^{-\frac{N}2} \end{array} \right)         \nonumber \\
                   &  \geqs \left( \begin{array}{c} \beta\Lambda -r^{2\delta-2}K_1\ttk_3 \beta^{N-1}(u+\omega)^{N-1}-K_2r^{2\delta-2}\phi_+^{-N-1}-c_n\rho\phi_+^{-N/2} \\ \beta \lambda + (d_n\tau -\frac{d_n}{n-1}\theta_-)\phi_+^{\frac{N}2} - \frac{d_n}{n-1}(((n-1)\tau+|\theta_-|/2)\psi^N)\phi_+^{-\frac{N}2} \end{array} \right).         \nonumber
 \end{align}
As in Theorem 4.1 in \cite{DIMM13}, the decay rate on $\Lambda$ ensures that we can choose $\beta$ sufficiently small
so that
$$
\frac{\beta\Lambda}{2} -r^{2\delta-2}K_1\ttk_3 \beta^{N-1}(u+\omega)^{N-1}>0.
$$
The smallness assumptions on 
$\sigma, \rho$, $J$ on $\cM$ and the smallness assumptions on $(2(n-1)\tau +|\theta_-|)\psi^N$ on $\Sigma$
imply that we can ensure that the first equation in the above array is nonnegative. 
For this fixed $\beta$, we observe that 
$$
\beta\lambda + (d_n\tau -\frac{d_n}{n-1}\theta_-)\phi_+^{\frac{N}2}>0.
$$
Therefore the smallness assumption on $(2(n-1)\tau +|\theta_-|)\psi^N$ implies
that the boundary equation can be made nonnegative as well.
Therefore, $\phi_+ = \beta(u+\omega)$ will be a global super-solution of the
Hamiltonian constraint with boundary condition \eqref{eq5:11july13}-\eqref{eq2:20dec13} if $\beta > 0$ is sufficiently small and
the conformal data satisfies the above assumptions.
\end{proof}

In the following Theorem,  we show that when $\psi = \phi_+ = \beta(u+\omega)$, where $u$ satisfies
\eqref{eq12:17oct13}, $\phi_+$ will be a global super-solution to the Hamiltonian constraint with boundary
conditions \eqref{eq5:11july13}-\eqref{eq2:20dec13} provided that $\beta> 0$ is chosen
sufficiently small and our data $\sigma, \rho, J$ is sufficiently small.  The significance of this result 
is that the supersolution acts as an {\em a priori} upper bound for the fixed point solution $\phi$, and therefore
$\phi \leqs \phi_+ = \psi$.  This implies that the
resulting fixed point $\phi$ will satisfy the marginally trapped surface conditions, which is why we refer
to the following supersolution construction as a {\bf marginally trapped surface supersolution}.

\begin{theorem}{\bf (Marginally Trapped Surface Supersolution for $g \in \cY^+$)}\label{thm1:24jan13}
Let the assumptions of Theorem~\ref{thm1:17oct13} hold with the exception of the
smallness assumption on $(2(n-1)\tau + |\theta_-|)\psi^N$, and let $u$ satisfy equation \eqref{eq12:17oct13}.
Then there exists a $\beta>0$ such that if $\phi_+ = \beta (u+\omega)$ and $\psi = \beta (u+\omega)$ on $\Sigma$, $\phi_+$ will
be a global supersolution to the Hamiltonian constraint with boundary condition \eqref{eq5:11july13}-\eqref{eq2:20dec13}
that also imposes the marginally trapped surface condition.
\end{theorem}
\begin{proof}
As in the proof of Theorem~\ref{thm1:17oct13},
we let $a_W = c_n|\sigma +\cL W|^2 \leqs 2|\sigma|^2 + 2|\cL W|^2$
and apply the estimate \eqref{eq1:23jan14}, where the constants $\ttk_1, \ttk_2, K_1$ and
$K_2$ are the same as in the previous proof.
We apply the Hamiltonian constraint \eqref{E:ham-abstract}
to $\phi_+$ and use the fact that $S(\nu,\nu) = ((n-1)\tau +|\theta_-|/2)\phi_+^N$ to obtain
\begin{align}
\cN(&\phi_+,W(\phi)) \nonumber \\
&= \left(\begin{array}{c}-\Delta \phi_+ +c_nR\phi_+ +b_n\tau^2\phi_+^{N-1} -a_{W}\phi_+^{-N-1}-c_n\rho\phi_+^{-N/2} \\ \partial_{\nu} \phi_+ + d_nH\phi_+ + (d_n\tau -\frac{d_n}{n-1}\theta_-)\phi_+^{\frac{N}2} - \frac{d_n}{n-1}((n-1)\tau+|\theta_-|/2)\phi_+^{\frac{N}2} \end{array}\right)\nonumber \\
                 &  \geqs \left( \begin{array}{c} \beta\Lambda  +\beta^{N-1} b_n\tau^2(u+\omega)^{N-1} -r^{2\delta-2}(K_1 (\phi^{\tiwedge})^{2N}+K_2)\phi_+^{-N-1}-c_n\rho\phi_+^{-N/2} \\ \beta \lambda +\left(\frac{d_n}{2(n-1)}|\theta_-|\right)\phi_+^{\frac{N}2} \end{array} \right).         \nonumber
 \end{align}
We now observe that 
\begin{align}\label{eq1:13dec13}
\ttk_2 &\leqs  4C^2\left(\|J\|_{L^p_{\delta-2}}+\|\sigma(\nu,\nu)\|_{1-\frac{1}{p},p;\Sigma}\right)^2\\
&\qquad +4C\left(\|((n-1)\tau+|\theta_-|/2)\|_{1-\frac{1}{p},p;\Sigma}\|\psi\|^{N-1}_{\infty}\|\psi\|_{s,p;\Sigma}\right)^2 \nonumber \\
&=C_1(\sigma,J)+\beta^{2N}C_2(\tau,\theta_-,u)  ,  \nonumber
\end{align}
and we have that
\begin{align}
K_2 \leqs 2r^{2-2\delta}(\sigma^2)^{\tiwedge} + C_1(\sigma,J)+\beta^{2N}C_2(\tau,\theta_-,u) = C_3(\sigma,J)+\beta^{2N}C_2(\tau,\theta_-,u). \nonumber
\end{align}
If $\ttk_3$ is as in the proof of Theorem~\ref{thm1:17oct13} and if we choose $\beta > 0$ sufficiently small so that
$$
\frac{\beta\Lambda}{2} - \left(\ttk_3K_1(u+\omega)^{-N-1}+C_2(\tau,\theta_-,u)\right)\beta^{N-1}r^{2\delta-2} > 0~~~\text{on}~~~\cM,
$$
then we can ensure $\phi_+$ will be a super-solution by imposing smallness assumptions on $\sigma, \rho$ and $J$ as in the proof of Theorem~\ref{thm1:17oct13}. 
Given that the second equation is positive for any choice of $\beta > 0$, $\phi_+ = \beta(u+\omega)$ will be a global super-solution of the
Hamiltonian constraint with boundary condition \eqref{eq5:11july13} for $\beta > 0$ sufficiently small and conformal data satisfying the assumptions of the Theorem.
\end{proof}

The following Lemma provides us with a method of constructing a global supersolution in the event that $g$
is not in the positive Yamabe class.  However, we require that the scalar curvature $R$ and boundary mean
curvature $H$ be bounded by functions of $\tau$ and $\theta_-$  on the sets where they are negative. 

\begin{lemma}\label{lem1:14feb14}
Suppose that $(\cM,g)$ is asymptotically Euclidean of class $W^{2,p}_{\gamma}$ with
$p>n$ and $\gamma\in (2-n,0)$.  Assume that $2-n< \delta < \gamma/2 $ and
that $\tau \in W^{1,p}_{\delta-1}$, $\rho \in L^p_{\gamma-2}, \sigma \in W^{1,2p}_{\gamma-1}, J \in {\bf L}^p_{\delta-2}, \theta_- \in W^{1-\frac1p,p}(\Sigma)$,
and $(2(n-1)\tau+|\theta_-|) >0$ on $\Sigma$.  Additionally assume that
\begin{itemize}
\item $-c_nR \leqs b_n\tau^2 $ on $\{x\in \cM:R(x) < 0\}$,

\item $-H \leqs (\tau +|\theta_-|/(n-1))$ on $\{x\in \Sigma: H(x) < 0\}$,

\end{itemize}
and that $A \in L^p_{\gamma-2}$ is nonnegative. 
Then for $\omega \in \cH$ with asymptotic limits $A_i \in [1,\infty)$,
 there exists a solution $\phi_A$ to the equation
\begin{align}\label{eq1:14feb14}
-&\Delta \phi_A = c_nA\phi_A^{-N-1}+c_n\rho\phi_A^{-\frac{N}{2}}, \\
&\partial_{\nu} \phi_A = \frac{d_n}{n-1}A\phi_A^{-\frac{N}{2}}, \nonumber
\end{align}
such that $\phi_A -\omega \in W^{2,p}_{\gamma}$.  If $W_{\phi}$ denotes the solution to the momentum constraint
for $\phi \leqs \phi_A$,
$A \geqs |\sigma+\cL W_{\phi}|^2$ on $\cM$, and $A \geqs S(\nu,\nu) = (\sigma(\nu,\nu)+\cL W_{\phi}(\nu,\nu))$ on $\Sigma$,
then $\phi_A$ will be a global supersolution of the Hamiltonian constraint with boundary conditions \eqref{eq5:11july13}-\eqref{eq2:20dec13}.
\end{lemma}
\begin{proof}
We note that $\phi_- =1 $  is a subsolution to \eqref{eq1:14feb14}. 
We obtain a supersolution by letting $\phi_+= \beta(u+1)$, where $\beta\gg1$ is sufficiently large and
$u$ is the solution to 
\begin{align}\label{eq1:15feb14}
-&\Delta u = c_nA+c_n\rho \\
&\partial_{\nu}u = \frac{d_n}{n-1}A. \nonumber
\end{align}
The maximum principle implies that $\phi_+> 0$ and for $\beta \gg1$, we have $\phi_- \leqs \phi_+$.
We may then apply Theorem~\ref{thm1:23jan14} to obtain a solution $\phi_A \geqs 1$ which tends to freely
specified $A_i \in [1,\beta]$ on each end $E_i$.  As $\beta$ can be arbitrarily large, the asymptotic limits $A_i$
can be freely specified numbers in $[1,\infty)$.  

Now we compute $\mathcal{N}(\phi_A,W)$ for $\phi \leqs \phi_A$, where $W = W(\phi)$ depends on $\phi$.  We obtain
\begin{align*}
\mathcal{N}(\phi_A,W) &= \left( \begin{array}{c} -\Delta\phi_A +c_nR\phi_A+b_n\tau^2\phi_A^{N-1}-c_n|\sigma+\cL W|^2\phi_A^{-N-1}-c_n\rho\phi_A^{-\frac{N}{2}}\\
\partial_{\nu}\phi_A + d_nH\phi_A +\left(d_n\tau+\frac{d_n}{n-1}|\theta_-|\right)\phi_A^{\frac{N}{2}}-\frac{d_n}{n-1}S(\nu,\nu)\phi_A^{-\frac{N}{2}} \end{array}\right)\\
&\geqs  \left( \begin{array}{c} c_nR\phi_A+b_n\tau^2\phi_A^{N-1}+c_n(A-|\sigma+\cL W|^2)\phi_A^{-N-1}\\
d_nH\phi_A +\left(d_n\tau+\frac{d_n}{n-1}|\theta_-|\right)\phi_A^{\frac{N}{2}}+\frac{d_n}{n-1}(A-S(\nu,\nu))\phi_A^{-\frac{N}{2}} \end{array}\right) \geqs 0,
\end{align*} 
where the last is expression is nonnegative given the assumptions on $A$, the fact that $\phi_A \geqs 1$, and 
the assumption that $-c_nR \leqs b_n\tau^2$ and $-H \leqs (\tau+|\theta_-|/(n-1))$ on the sets where
$R$ and $H$ are negative.
\end{proof}

The previous Lemma tells us that $\phi_A$ will be a global supersolution to the Hamiltonian constraint
with boundary conditions \eqref{eq5:11july13} provided that we can find a function $A$ which bounds
the terms in $\cN(\phi,W(\phi))$ that depend on $W(\phi)$.  We construct such a function in the following
theorem and show that when $\psi =  \phi_A$ the function $\phi_A$ is a marginally trapped surface
supersolution.

\begin{theorem}{\bf (Global Supersolution for bounded $R$ and $H$)}\label{thm1:15feb14}
Let the assumptions of Lemma~\ref{lem1:14feb14} hold and suppose that
$\bV$ satisfies the conditions of Theorem~\ref{thm1:12feb14}.  Additionally assume
that $\Sigma$ is compact.
Then there
exists an $\e >0$ such that if
\begin{align}\label{eq2:15feb14}
\|\nabla\tau\|_{L^p_{\delta-2}} < \e \quad \text{and} \quad \|\tau + |\theta_-|\|_{W^{1-\frac{1}{p},p}(\Sigma)} < \e,
\end{align}
and $A = Cr^{2\delta-2}$ for some constant $C>0$, the solution $\phi_A$ to \eqref{eq1:14feb14} will be a 
global supersolution to the Hamiltonian constraint with boundary conditions \eqref{eq5:11july13}-\eqref{eq2:20dec13}.
Moreover, if $\phi_A = \psi$ the marginally trapped surface condition will hold.
\end{theorem}
\begin{proof}
Lemma~\ref{lem1:14feb14} implies that $\phi_+ = \phi_A$ will be a supersolution for $\phi \leqs \phi_A$ provided that we can choose $A \geqs |\sigma+\cL W|^2$ on $\cM$
and $A \geqs (\sigma(\nu,\nu)+\cL W(\nu,\nu))$ on $\Sigma$.  Using the estimate from
Proposition~\ref{prop1:5dec13}, we have
$$
|\sigma+\cL W|^2 \leqs 2|\sigma|^2 +2|\cL W|^2 \leqs 2|\sigma|^2 +2r^{2\delta-2}(\ttk_1\|\phi_A\|_{\infty}^{2N}+\ttk_2),
$$
where 
\begin{align*}
&\ttk_1 = 2 C_1^2 \|\nabla \tau\|_{L^p_{\delta-2}}^2,\\
&\ttk_2 = 2C_2^2\left(\|J\|_{L^p_{\delta-2}}+\|\sigma(\nu,\nu)\|_{1-\frac{1}{p},p;\Sigma} \right.\\
&\qquad + \left. \|((n-1)\tau+|\theta_-|/2)\|_{1-\frac{1}{p},p;\Sigma}\|\psi\|^{N-1}_{\infty}\|\psi\|_{s,p;\Sigma}\right)^2. \nonumber
\end{align*}
Setting $A = Cr^{2\delta-2}$, if we can find $C>0$ so that
\begin{align}\label{eq1:18feb14}
2\|\sigma\|^2_{L^{\infty}_{\delta-1}} +2(\ttk_1\|\phi_A\|_{\infty}^{2N}+\ttk_2) &\leqs C, \\
|S(\nu,\nu)| \leqs (2(n-1)\tau+|\theta_-|)\psi^N &\leqs C\min_{x\in \Sigma}(r^{2\delta-2}), \nonumber
\end{align}
the conditions of Lemma~\ref{lem1:14feb14} will be satisfied and $\phi_A$ will be a global supersolution.  For arbitrary an $\psi \in W^{1-\frac{1}{p},p}(\Sigma)$
that is independent of $A$, we choose 
$$
C > \max\{2(\|\sigma\|^2_{L^{\infty}_{\delta-1}}+\ttk_2), \alpha \max_{x\in \Sigma} (2((n-1)\tau+|\theta_-|)\psi^N)\},
$$
where $\alpha = 1/(\min_{x\in \Sigma}(r^{2\delta-2}))$.  Taking $\ttk_1 = \|\nabla \tau\|_{L^p_{\delta-2}}^2$
to be sufficiently small we can ensure that both inequalities in \eqref{eq1:18feb14} hold.

To obtain our marginally trapped supersolution, we set $\psi = \phi_A$.  In this case we take
$$
C> 2(\|\sigma\|^2_{L^{\infty}_{\delta-1}}+\ttk_2),
$$
and then require that both $\|\nabla \tau\|_{L^p_{\delta-2}}^2$ and $\|((n-1)\tau+|\theta_-|/2)\|_{1-\frac{1}{p},p;\Sigma}$
be sufficiently small to obtain the inequalities in \eqref{eq1:18feb14}.
\end{proof}

The final two theorems of this section provide us with a method to construct global subsolutions $\phi_- \le \phi_+$, where
$\phi_+$ is any of the supersolutions
constructed in Theorems~\ref{thm1:17oct13}, \ref{thm1:24jan13}, or \ref{thm1:15feb14}.

\begin{theorem}{\bf (Global Subsolution for $g \in \cY^+$) }\label{thm2:17oct13}
Suppose that $(\cM,g)$ is asymptotically Euclidean of class $W^{2,p}_{\gamma}$, with $n < p$
and $ \gamma \in (2-n, 0)$.  Additionally assume that $\cY_g > 0$, $2-n < \delta < \gamma/2$,
$\tau \in W^{1,p}_{\delta-1}$, $\rho \in L^p_{\gamma-2}, \sigma \in W^{1,2p}_{\gamma-1}, J \in {\bf L}^p_{\delta-2}, \theta_- \in W^{1-\frac1p,p}(\Sigma)$,
and $((n-1)\tau + |\theta_-|) > 0$ on $\Sigma$. Then there exists a 
subsolution $\phi_->0$ to the Hamiltonian constraint with boundary conditions \eqref{eq5:11july13}-\eqref{eq2:20dec13} such that $\phi_- - \alpha\omega \in W^{2,p}_{\gamma}$
for $\alpha > 0$ sufficiently small .
\end{theorem}
\begin{proof}
Because $g \in \cY^+$, there exists $u \in W^{2,p}_{\gamma}$ which solves
\begin{align}\label{eq1:6dec13}
-&\Delta u + (c_n R+b_n\tau^2) u = -\omega (c_nR+b_n\tau^2), \\
&\partial_{\nu} u + \left(d_n H+d_n\tau + \frac{d_n}{(n-1)}|\theta_-| \right) u = -\omega\left(d_n H+d_n\tau + \frac{d_n}{(n-1)}|\theta_-|\right). \nonumber 
\end{align}
Set $\phi_- = \alpha(u+\omega)$, where $\alpha>0$ will be determined.  We observe that
\begin{align*}
-&\Delta \phi_- + (c_n R+b_n\tau^2) \phi_- = 0, \\
&\partial_{\nu} \phi_- + \left(d_n H+d_n\tau + \frac{d_n}{(n-1)}|\theta_-| \right)  \phi_- = 0, \nonumber 
\end{align*}
and by the maximum the principles \ref{wmaxprinc} and \ref{smaxprinc}, $\phi_- > l > 0$ given that
$\phi_- \to A_j > 0$ on each end. 

We claim that for $\alpha$ sufficiently small, $\phi_- = \alpha \psi$ is a global subsolution. 
Suppose that $\phi \geqs \phi_-$.  Then we have
\begin{align*}
\mathcal{N}(&\phi_-,S(\phi)) \\
&\leqs \left( \begin{array}{c}  b_n(\alpha^{N-1} (u+\omega)^{N-1}- \alpha (u+\omega))\tau^2 - c_n|\sigma+\cL W|^2\phi_-^{-N-1} - c_n\rho \phi_-^{-\frac{N}{2}} \\ \left( \alpha^{\frac{N}{2}}(u+\omega)^{\frac{N}{2}}-\alpha (u+\omega)\right) \left(d_n\tau + \frac{d_n}{(n-1)}|\theta_-| \right)  \end{array}\right),  
\end{align*}
where we have used the fact that $S(\nu,\nu) = ((n-1)\tau +|\theta_-|/2)\psi^N > 0$. We observe that if we take $\alpha$ sufficiently small, both expressions in the above array will be nonpositive.  
\end{proof}

\begin{theorem}{\bf(Global Subsolution for bounded $R$ and $H$)}\label{thm1:18feb14}
Let the assumptions of Lemma~\ref{lem1:14feb14} hold along with additional assumption that $S(\nu,\nu) \geqs 0$.  Then there exists a solution $u$ to
\begin{align}\label{eq3:15feb14}
-&\Delta u +c_nRu+b_n\tau^2u^5 = 0 \quad \text{on $\cM$}, \\
&\partial_{\nu} u + d_nHu +\left(d_n\tau+\frac{d_n}{n-1}|\theta_-|\right)u^{\frac{N}{2}} =  0 \quad \text{on $\Sigma$}, \nonumber
\end{align}
such that $u-\omega \in W^{2,p}_{\gamma}$.  Moreover, for any $\alpha \in (0,1)$ the function $\phi_- = \alpha u$
will be a global subsolution to the Hamiltonian constraint with boundary conditions \eqref{eq5:11july13}-\eqref{eq2:20dec13}.  
\end{theorem}
\begin{proof}
We observe that $u_- \equiv 0$ is a subsolution of \eqref{eq3:15feb14}
and $u_+ \equiv \beta \ge 1$ is a supersolution given the assumptions on $R$ and $H$.  Let
$\omega$ have asymptotic limits $A_i \in (0,\infty)$.  By Theorem~\ref{thm1:23jan14} we can choose $\beta$ large enough so that Eq.~\eqref{eq3:15feb14} 
has a solution $u$ such that $u-\omega \in W^{2,p}_{\gamma}$.  By construction, $u \geqs 0$.  We note that if
$u(x_0) = 0$ for some $x_0 \in \cM$, then $x_0$ will be a minimum of $u$.
Both $u$ and $0$ satisfy the elliptic equation
\begin{align*}
-&\Delta v + (c_nR+b_n\tau u^{N-2}) v = 0 \quad \text{on $\cM$},\\
&\partial_{\nu}v + (d_n H+ \left(d_n\tau+\frac{d_n}{n-1}|\theta_-|\right)u^{\frac{N-2}{2}})v = 0 \quad \text{on $\Sigma$} ,
\end{align*}
 and $u$ and the zero function will coincide
up to first order at $x_0$.  Alexandrov's Theorem (cf. \cite{yCBjIjY00}) implies $u \equiv 0$, which
contradicts the fact that $u \to A_i > 0$ on each $E_i$.  So $u > 0$ on $\cM$.

Setting $\phi_- = \alpha u $ for $\alpha \in (0,1)$, we calculate $\cN(\phi_-,S(\phi))$ for $\phi \geqs \phi_-$:
\begin{align*}
\mathcal{N}(\phi_-,S(\phi))
= \left( \begin{array}{c}  - c_n|\sigma+\cL W|^2\phi_-^{-N-1} - c_n\rho \phi_-^{-\frac{N}{2}} \\ -\frac{d_n}{n-1}S(\nu,\nu)\phi_-^{-\frac{N}{2}}  \end{array}\right) \leqs 0.  
\end{align*}
Therefore $\phi_- = \alpha u $ is a global subsolution to the Hamiltonian constraint.  
\end{proof}

Given $\omega_1 \in \cH$ which tends to positive values on each end $E_i$, for arbitrarily small $\alpha >0$ we may obtain a positive subsolution $\phi_-$ such
that $\phi_-  -\alpha\omega_1 \in W^{2,p}_{\gamma} $.  Similarly, given $\omega_2 \in \cH$ which
tends to positive values on each end, for $\beta>0$ sufficiently small there exists a positive supersolution $\phi_+$ such that $\phi_+ - \beta \omega_2 \in W^{2,p}_{\delta}$. 
By choosing $\alpha \ll \beta$, we can ensure that $\alpha \omega_1$ is asymptotically bounded by $\beta\omega_2$ and that $\phi_- \leqs \phi_+$.
Now that we have constructed barriers for the Hamiltonian constraint
with the specified boundary conditions, we are ready to prove Theorem~\ref{thm1:2dec13}.

\section{Non-CMC Solutions: Fixed Point Argument}
\label{sec:proof}

Given a set of global barriers $\phi_- \leqs \phi_+$, which we derived in Section~\ref{sec:barriers}, Theorem~\ref{thm1:2dec13} will follow by 
using a variation of the fixed point argument first developed in \cite{HNT07b}.  The following argument closely follows the work done in \cite{DIMM13},
where the authors extended the argument in \cite{HNT07b} to AE manifolds with no boundary. 
We slightly modify this fixed point argument to include our boundary problem. 

 Before we prove our fixed point theorem, we first  
 discuss the properites of the solution map of the Hamiltonian constraint with the associated marginally trapped surface boundary conditions.
 In particular, we show that this map is well-defined up to the asymptotic limit of the solution, and then show that it is continuous.
 
 Let $W \in W^{2,p}_{\delta}$ be a given vector field with $2-n < \delta \le \gamma/2< 0$, and let $\phi_- \leqs \phi_+$ be sub-and supersolutions of $\cN(\phi,W)$.  
 For a given $k$-tuple $A_1,\cdots,A_k $ of positive, real numbers, let $\omega \in \cH$
 be the associated harmonic function.
 By Theorems~\ref{thm1:23jan14} and \ref{uniq}, for a given $W$, sub- and supersolutions $\phi_-$ and $\phi_+$, and $\omega$
 that is asymptotically bounded by $\phi_-$ and $\phi_+$,
 there exists a unique solution to $\cN(\phi,W) = 0$ such that $\phi- \omega \in W^{2,p}_{\gamma}$.  
 Therefore, for given a $W, \phi_-\le \phi_+$, and $\omega$,
 we define  $T(W) = \phi$ to be the solution map giving this unique solution. 
  
 Given that $T$ is used to construct our fixed point map for the Schauder Theorem, we require that the $T(W)$ 
 be a continuous mapping.  We note that ${\mathcal{G}(\phi) = {\it i}( T( \cS(\phi)))}$, where ${{\it i}:\cH+ W^{2,p}_{\gamma} \to C^0}$ is the compact embedding \eqref{embed}
 and $\cS$ is the continuous solution map
 of the momentum constraint.  Therefore the continuity of $T$ will imply
 the continuity of $\mathcal{G}$.  We set $\beta(W) = \sigma + \cL W$ and define
 $\mathcal{L}(\beta(W) ) = T(W)$.  Then for fixed data $(g,\tau,\rho, \theta_-)$, $\mathcal{L}(\beta)$ is the solution map of the Lichnerowicz equation
 with boundary conditions~\eqref{eq5:11july13} for a given 2-tensor $\beta$.  That is, $\cL(\beta)$ gives the solution of
 \begin{align*}
 -&\Delta \phi + c_n R \phi + b_n \tau^2\phi^{N-1}-c_n|\beta|^2\phi^{-N-1} - c_n \rho\phi^{-\frac{N}{2}} = 0 ~~~~\text{on $\cM$}, \\
&\partial_{\nu}\phi+d_nH\phi+\left(d_n \tau - \frac{d_n}{n-1}\theta_- \right)\phi^{\frac{N}{2}}-\frac{d_n}{n-1}\beta(\nu,\nu)\phi^{-\frac{N}{2}} = 0 ~~~~ \text{on $\Sigma$}. 
 \end{align*}
 To prove the continuity of $T$ it is sufficient to prove the continuity of
 $\mathcal{L}$ in $\beta$.  The proof is based on the Implicit Function Theorem argument developed
 in  \cite{dM09}.
 
 \begin{proposition}\label{prop1:18dec13}
 Suppose $(\cM,g)$ is asymptotically Euclidean of class $W^{2,p}_{\gamma}$, with $\gamma \in (2-n,0)$ and $2 > \frac{n}{p}$.  
 Additionally assume that $\tau \in W^{1,p}_{\gamma/2-1}$, $\rho \in L^p_{\gamma-2}$, $\theta_- \in W^{1-\frac{1}{p},p}(\Sigma)$,
 and $\beta \in W^{1,2p}_{\gamma/2-1}$.  If 
 $((n-1)\tau+|\theta_-|) \geqs 0$ and $ \beta_0(\nu,\nu) \geqs 0, $
 then $\mathcal{L}$ is a $C^1$ map from $W^{1,2p}_{\gamma/2-1}$ to $W^{2,p}_{\gamma}$. 
 \end{proposition}
 \begin{proof}
 As in \cite{dM09}, we exploit the conformal covariance of the Lichnerowicz equation. Let $\hat{g} = \phi^{N-2}g$ and $\hat{\mathcal{L}}$
 be the solution map associated with $\hat{g}$.  By the conformal covariance of the boundary problem demonstrated in \cite{HoTs10a}, we have that
 $$
 \hat{\mathcal{L}}(\hat{\beta}) = \phi^{-1} \mathcal{L}(\beta) = 1, \quad \text{where $\hat{\beta} = \phi^{-2N}\beta$.}
$$
Therefore it suffices to demonstrate the continuity of $\mathcal{L}$ in a neighborhood of $\beta_0$ such that $\mathcal{L}(\beta_0) = 1$,
and we may drop the hat notation.

Define
\begin{align}\label{eq3:17dec13}
\mathcal{F}(\phi,\beta) = \left[ \begin{array}{c}    -\Delta \phi + c_n R \phi + b_n \tau^2\phi^{N-1}-c_n|\beta|^2\phi^{-N-1} - c_n \rho\phi^{-\frac{N}{2}} \\
							 \partial_{\nu}\phi+d_nH\phi +\left(d_n\tau +\frac{d_n}{(n-1)}|\theta_-|\right)\phi^{\frac{N}2} - \frac{d_n}{(n-1)}\beta(\nu,\nu)\phi^{-\frac{N}{2}} \\
                                 \end{array}\right].
\end{align}
It is clear that $\mathcal{F}(\mathcal{L}(\beta),\beta) = 0$, and a standard computation shows that the Gateaux derivative is given by
\begin{align}
\mathcal{F}'_{\phi,\beta}(h,K) = \left( \begin{array}{c}   -\Delta h + \alpha_1(\phi,\beta) h  - 2c_n\phi^{-N-1}\beta\cdot K \\
									   \partial_{\nu} h  + \alpha_2(\phi,\beta) h  - \frac{d_n}{(n-1)}\phi^{-\frac{N}2}K(\nu,\nu)\\
              \end{array} \right),
\end{align}
where 
$$
\alpha_1(\phi,\beta) = c_n R + (N-1)b_n\tau^2\phi^{N-2} + (N+1)c_n|\beta|^2\phi^{-N-2}+\frac{c_nN}{2}\rho\phi^{-\frac{N}{2}-1},
$$
and
$$
\alpha_2(\phi,\beta) = d_n H +\frac{N}{2}\phi^{\frac{N}{2}-1}\left(d_n\tau+\frac{d_n}{(n-1)}|\theta_-|\right)+\frac{Nd_n}{2(n-1)}\phi^{-\frac{N}{2}-1}\beta(\nu,\nu).
$$
The multiplication properties of weighted Sobolev spaces imply that the operator $\mathcal{F}'$ is continuous in $\phi$ and $\beta$.
We have
\begin{align*}
\mathcal{F}'_{1,\beta_0}(h,0) = \left( \begin{array}{c}   -\Delta h +  \left(c_n R + (N-1)b_n\tau^2 + (N+1)c_n|\beta_0|^2+\frac{c_nN}{2}\rho\right) h  \\
									   \partial_{\nu} h  +\left( d_n H +\frac{N}{2}\left(d_n\tau+\frac{d_n}{(n-1)}|\theta_-|\right)+\frac{Nd_n}{2(n-1)}\beta_0(\nu,\nu)  \right) h \\
              \end{array} \right),
\end{align*}
and given that $\mathcal{F}(1,\beta_0) ={\bf 0}$, 
\begin{align*}
   &c_n R+ b_n \tau^2-c_n|\beta_0|^2 - c_n \rho =0, \nonumber\\
&d_nH+\left(d_n\tau +\frac{d_n}{(n-1)}|\theta_-|\right) - \frac{d_n}{(n-1)}\beta_0(\nu,\nu)=0. \nonumber
\end{align*}
This implies that
\begin{align}
\mathcal{F}'_{1,\beta_0}(h,0) = \left( \begin{array}{c}   -\Delta h +  \left( (N-2)b_n\tau^2 + (N+2)c_n|\beta_0|^2+\frac{N+2}{2}c_n\rho\right) h  \\
									   \partial_{\nu} h  + \left(\frac{N-2}{2}\left(d_n\tau+\frac{d_n}{(n-1)}|\theta_-|\right)+\frac{(N+2)}{2}\frac{d_n}{(n-1)}\beta_0(\nu,\nu)  \right) h \\
              \end{array} \right).
\end{align}
The assumptions $\beta_0(\nu,\nu) \geqs 0$ and $((n-1)\tau + |\theta_-|) \geqs 0$ imply that ${\mathcal{F}'_{1,\beta_0}:W^{2,p}_{\gamma} \to L^p_{\gamma-2}}$ is 
an isomorphism, and the Implicit Function Theorem implies that $\mathcal{L}$ is continuous in a neighborhood of $\beta_0$.
  \end{proof}
 
 Now that we have established existence of global barriers and showed that $\mathcal{G}$ is continuous, we
 are ready to prove Theorems~\ref{thm1:2dec13} and \ref{thm2:2dec13}.
 
 \bigskip
 
{\bf Proof of Theorem~\ref{thm1:2dec13}}
Let $\mathcal{C}^0_+$ denote the set of strictly positive bounded functions on $\cM$. If $\phi \in \mathcal{C}^0_+$, then
by Proposition \ref{prop1:14oct13}, the vector field $W = \mathcal{S}(\phi) \in W^{2,p}_{\delta}$ given
by the solution map of the momentum constraint with boundary conditions~\eqref{eq5:11july13} is well-defined.  
By the remarks preceding Proposition~\ref{prop1:18dec13} and Theorems~\ref{thm1:23jan14} and \ref{uniq}, given $W \in W^{2,p}_{\delta}$,
sub-and super-solutions $\phi_-\leqs \phi_+$, and a harmonic function $\omega$ as in Proposition~\ref{prop1:23dec13} that is asymptotically bounded by $\phi_- $ and $\phi_+$, 
the solution map $T(W) = \varphi$ is well-defined and continuous.

Let $\omega_1,\omega_2 \in \cH$ tend to positive real numbers on each end and suppose that $\omega_2 $ is asymptotically bounded above
by $\omega_1 $.
Let $\phi_+$ be the global supersolution obtained from either Theorem~\ref{thm1:17oct13} or \ref{thm1:24jan13}, where $\phi_+-\beta\omega_1 \in W^{2,p}_{\gamma}$.  
Note that we use the supersolution from Theorem~\ref{thm1:17oct13} if we wish to solve the coupled constraints
with boundary conditions~\eqref{eq5:11july13}-\eqref{eq2:20dec13} with arbitrary $\psi$.  If we wish to obtain a solution $\phi$ satisfying the marginally trapped surface condition
$\phi \leqs \psi$, then we use the supersolution from Theorem~\ref{thm1:24jan13}.
Let $\phi_- \leqs \phi_+$ be the global subsolution
obtained from Theorem~\ref{thm2:17oct13}, where $\phi_- - \alpha\omega_2 \in W^{2,p}_{\gamma}$.  
Let $\omega \in \cH$ be asymptotically bounded by $\phi_- \leqs \phi_+$.  
For this choice of sub-and supersolutions and $\omega$ we may apply Theorem~\ref{thm1:23jan14} to 
obtain $\varphi = T(W)$.  Following the proof of Theorem~\ref{thm1:23jan14}, $\varphi = \omega + \hat{\varphi} \in \cH+W^{2,p}_{\gamma}$.
Let $\it{i}$ denote the compact inclusion $\mathcal{H} + W^{2,p}_\delta \hookrightarrow \mathcal{C}^0$.
A solution $(\phi, W)$ to \eqref{E:ham-abstract} then corresponds to a fixed point of the mapping $\mathcal{G}(\phi) = \it{i}(T(\mathcal{S}(\phi)))$,
which is a continuous, compact mapping.

Define the bounded convex set $\mathcal{A}: = \{\phi \in \mathcal{C}^0_+ : \phi_-\leqs \phi \leqs \phi_+\}$. By construction,
$\mathcal{G}$ maps $\mathcal{A}$ to itself.  Moreover, $\mathcal{A}$ is closed, bounded, and convex. Therefore the Schauder fixed point theorem
implies that $\mathcal{A}$ contains a fixed point $\phi$ of $\mathcal{G}$. Standard estimates imply that $\phi$ 
and $W(\phi)$ both have the desired regularity.  \qed

\bigskip

{\bf Proof of Theorem~\ref{thm1:12feb14}}
The proof if the same as the proof of Theorem~\ref{thm1:2dec13} except for the barriers used.  Given $A_i \in [1,\infty)$, let
$\omega \in \mathcal{H}$ be the associated harmonic function.  Choose $\omega_1, \omega_2  \in \mathcal{H}$ such that
$\omega_1 \leqs \omega_2$ and $\omega$ is asymptotically bounded by $\omega_1$ and $\omega_2$.  By Theorem~\ref{thm1:15feb14}
there exists a global supersolution $\phi_+ = \phi_A$ to the Hamiltonian constraint with boundary conditions~\eqref{eq5:11july13}-\eqref{eq2:20dec13}
such that $\phi_+ - \omega_2 \in W^{2,p}_{\gamma}$.  By setting $\psi = \phi_+$ we can impose the marginally trapped surface condition $\phi \leqs \phi_+$.
By Theorem~\ref{thm1:18feb14} we obtain a global subsolution $\phi_- \leqs \phi_+$ such that $\phi_- - \alpha\omega_1\in W^{2,p}_{\gamma}$
for $\alpha \in (0,1)$.  By construction, the function $\omega$ is asymptotically bounded by $\phi_- \leqs \phi_+$, and as in the proof
of Theorem~\ref{thm1:2dec13} we apply Theorem~\ref{thm1:23jan14} to obtain a solution $\varphi$ such that $\varphi- \omega \in W^{2,p}_{\gamma}$.
Therefore, for given asymptotic limits $A_i \in [1,\infty)$, the solution map $\varphi= T(W(\phi))$ is well-defined for $\phi_- \leqs \phi \leqs \phi_+$.
The rest of the proof follows from the arguments made in the proof of Theorem~\ref{thm1:2dec13}.

\section{Near-CMC Solutions: An Implicit Function Theorem Argument}
\label{sec:implicit}

In this section, we provide an alternative approach to obtain solutions to the conformal equations satisfying the
marginally trapped surface boundary conditions.  This approach is based on the Implicit Function Theorem 
argument given in \cite{yCBjIjY00}, and therefore requires that $\|\tau\|_{W^{1,p}_{\delta-1}}$ be sufficiently small.  

We first recall the Implicit Function Theorem.  Suppose that $U$ and $V$ are open subsets of Banach spaces $X$ and $Y$
and $\mathcal{F}$ is a $C^1$ mapping from $U\times V$ into a Banach space $Z$:
$$
\mathcal{F}: X\times Y \to Z.
$$
The Implicit function theorem states that if $\mathcal{F}_y(x_0,y_0)$ is invertible at some solution of $\mathcal{F}(x_0,y_0) = 0$,
then there exists a neighborhood $U'\times V' \subset U\times V$ of $(x_0,y_0)$ such that for each $x \in U'$, there exists a 
unique $y \in V'$ such that $\mathcal{F}(x,y)=0$.  That is, there exists an invertible function $\rho: V' \to U'$ such that
all solutions to $\mathcal{F}(x,y) = 0$ in $U'\times V'$ are of the form $F(\rho(y),y) =0$.  Moreover, if $\mathcal{F}$ is $C^1$ then
$\rho$ is $C^1$.

Given a $k$-tuple $A_1,\cdots,A_k $ of positive numbers in $\mathbb{R}$, let $\omega \in \cH$ be the associated harmonic function.
Suppose that $2-n < \delta \leqs \gamma/2 < 0$.  As in \cite{yCBjIjY00}, define the variables
\begin{align}
x &= (\tau, J) \in X = W^{1,p}_{\delta-1} \times L^p_{\delta-2}(T\cM),   \\
y &= ( \phi-\omega,W) \in Y =(  W^{2,p}_{\gamma}\times W^{2,p}_{\delta}(T\cM) ) \cap \{\phi > 0\}, \nonumber\\
Z & = L^{p}_{\gamma-2}\times W^{1-\frac1p,p}(\Sigma) \times L^p_{\delta-2}(T\cM)\times W^{1-\frac1p,p}(T\Sigma). \nonumber
\end{align}
Let
\begin{align}\label{eq1:17dec13}
\mathcal{F}(x,y) = \left[ \begin{array}{c}    -\Delta \phi + c_n R \phi + b_n \tau^2\phi^{N-1}-c_n|\sigma + \cL W|^2\phi^{-N-1} - c_n \rho\phi^{-\frac{N}{2}} \\
							 \partial_{\nu}\phi+d_nH\phi \\
                                                                        \Delta_{\mathbb{L}} W  + \frac{n-1}{n} \nabla \tau \phi^N + J \\
                                                                        \cL W(\nu,\cdot) - \bV(\phi,\tau)
                                \end{array}\right],
\end{align}
where $\bV(\phi,\tau)(\nu) = (n-1)\tau\phi^{N}-\sigma(\nu,\nu)$, which implies that $S(\nu,\nu) = (n-1)\tau\phi^N$.  We observe
that solutions to $\mathcal{F}(x,y) = 0$ represent solutions to the coupled system \eqref{eq3:27nov13}-\eqref{eq4:27nov13} 
with boundary conditions \eqref{eq5:11july13}-\eqref{eq2:20dec13}
when $|\theta_-| = 0$ and $\psi = \phi$.  These solutions will satisfy the marginally trapped surface conditions if $\tau \geqs 0$ on $\Sigma$.  

Fix $\sigma \in W^{1,2p}_{\gamma/2-1}$ and $\rho \in L^p_{\gamma-2}$. 
In order to apply the Implicit Function Theorem to \eqref{eq1:17dec13}, 
we require that $\bV = \bV(\phi,\tau)$ be a $C^1$ vector field in $(\phi, \tau)$. 
We construct such a $\bV$ in the following proposition.

\begin{proposition}\label{prop1:20dec13}
Suppose that $(\cM,g)$ is an $n$-dimensional asymptotically Euclidean manifold of class $W^{2,p}_{\gamma}$ with $\gamma\in (2-n,0)$ 
and compact boundary $\Sigma$.  If $2-n < \delta <\gamma/2 $, $\tau\in  W^{1,p}_{\delta-1}$, and $\phi \in C^0$, then 
there exists a vector field $ \bV(\phi,\tau) \in \bW^{1,p}$ that is $C^1$ in $\phi$ and $\tau$.  Moreover, $\bV$ satisfies
\begin{align}
\bV(\nu) = (n-1)\tau\phi^N-\sigma(\nu,\nu) .
\end{align}
\end{proposition}
\begin{proof}
For every point $p \in \Sigma$ there is a neighborhood $U$
and a coordinate map $\Psi$ such that $\Psi(p) \in \{{\bf x}\in\mathbb{R}^n~|~x_n = 0\}$ and $V = \Psi(U) \subseteq \{{\bf x} \in\mathbb{R}^n~|~x_n \geqs 0\}$.
There exists a radius $R> 0$ such that $B_R(\Psi(p))\cap \{\bx: x_n \geqs 0\}  \subset V$.  Let $$A = \Psi^{-1}(B_R(\Psi(p))\cap \{\bx:x_n \geqs 0\} ).$$
On $V$, we define the constant vector field $\bW({\bf x}) = (0,\cdots,0,-1)$, and then consider the pullback $\bX= \Psi^*(\bW)$ on $A$.  By construction, $\bX= \nu$ on $\Sigma \cap A$.  
The compactness of $\Sigma$ implies that there exists some collection of $p_i \in \Sigma$ 
such that the associated sets $A_i$ as above determine a finite covering of $\Sigma$ for $1 \leqs i \leqs M$.  
Let $\bX_i$ be the associated local vector fields defined on $A_i$, and let $A_0$ be an open set
such that $\cup_{i=0}^M A_i = \cM$.  Finally, let $\chi_i$ be a partition of unity subordinate to
the covering $\{A_i\}$.  Setting 
$$\bV = \sum_{i=0}^M \chi_i ((n-1)\tau\phi^N-\sigma(\nu,\nu))\bX_i =  ((n-1)\tau\phi^N-\sigma(\nu,\nu))\bX,$$ 
where $\bX_0 = {\bf 0}$ on $A_0$,
it is clear that $\bV =  ((n-1)\tau\phi^N-\sigma(\nu,\nu))\nu$ on $\Sigma$ given that $\bX = \nu$ on $\Sigma$.  
By construction, $\bV\in \bW^{1,p}$ given the regularity assumptions on $\tau,\sigma$ and $\phi$ and 
the fact that $\bV$ vanishes outside of a neighborhood of $\Sigma$.  Clearly $\bV$ will be $C^1$ in $\phi$ and $ \tau$.  
\end{proof}

The properties of the vector $\bV$ constructed in Proposition~\ref{prop1:20dec13} and the multiplication properties 
of weighted Sobolev spaces imply that $\mathcal{F}(x,y):X\times Y \to Z$ will be $C^1$ as long as
$p > n$ and $\delta \in (2-n,0)$ (cf. \cite{Ba86,yCBjIjY00}).
Letting $\bV = ((n-1)\tau\phi^N-\sigma(\nu,\nu))\bX$ as in the
proof of Proposition~\ref{prop1:20dec13}, the partial derivative for a given $(x,y)$ 
\begin{align*}
\mathcal{F}'_y(x,y):Y &\to Z,\\
(h,\beta): &\to \mathcal{F}'_y(x,y)(h,\beta),
\end{align*}
is given by
\begin{align}
\mathcal{F}'_y(x,y)(h,\beta) = \left( \begin{array}{c}   -\Delta h + \alpha(\phi,W) h  - 2c_n\phi^{-N-1}(\sigma+ \cL W)\cdot\cL \beta \\
									   \partial_{\nu}h+ d_n H h  \\
									    \Delta_{\mathbb{L}} \beta + N\frac{(n-1)}{n}\nabla\tau \phi^{N-1} h \\
									    \cL\beta(\nu,\cdot) - N(n-1)\tau\phi^{N-1}\bX h
              \end{array} \right),
\end{align}
where 
$$
\alpha(\phi,W) = c_n R + (N-1)b_n\tau^2\phi^{N-2} + (N+1)c_n|\sigma+\cL W|^2\phi^{-N-2}+\frac{c_nN}{2}\rho\phi^{-\frac{N}{2}-1}
$$
and $\bX =\nu$ is a smooth vector field on $\cM$ vanishing in a neighborhood of $\Sigma$.

\begin{theorem}\label{thm1:17dec13}
Suppose that $(\cM,g)$ is an n-dimensional asymptotically Euclidean manifold of class $W^{2,p}_{\gamma}$, where $\gamma \in (2-n, 0)$ and $p > n$.  
Assume that $\sigma \in W^{1,2p}_{\gamma/2-1}$ and $\rho \in L^{p}_{\gamma-2}$ are given.  Suppose that $\mathcal{F}(x,y) =0$ has a solution when $y_0 = (\tau_0,{\bf J}_0) = (0,{\bf 0})$,
and denote this solution by $x_0 = (\phi_0,W_0)$.  If $\alpha(\phi_0,W_0) \geqs 0$ and $H\geqs 0$, then there exists a neighborhood $U$ of $(\tau_0,{\bf J}_0) $ in $X$ such 
that the coupled constraints with boundary conditions ~\eqref{eq5:11july13}-\eqref{eq4:8aug13} have a unique solution $(\phi,W)$, $\phi > 0$, $(\phi-\omega,W) \in Y$.
\end{theorem}
\begin{proof}
We calculate
 \begin{equation}\label{eq3:20dec13}
 \begin{aligned}
 \mathcal{F}'_y(x_0,y_0)(h,\beta) = \left( \begin{array}{c}   -\Delta h + \alpha(\phi_0,W_0) h  - 2c_n\phi_0^{-N-1}(\sigma+\cL W_0)\cdot\cL \beta \\
									   \partial_{\nu} h +d_nH h  \\
									    \Delta_{\mathbb{L}} \beta  \\
									    \cL\beta(\nu,\cdot) 
              \end{array} \right).
              \end{aligned}
 \end{equation}
By Proposition~\ref{prop1:14oct13} the operator $\mathcal{F}'_y(x_0,y_0):Y \to Z$ is invertible.  Therefore the Implicit Function Theorem implies the result.
\end{proof}

\begin{corollary}\label{cor1:19feb14}
Suppose that $(\cM,g)$ is an n-dimensional asymptotically Euclidean manifold of class $W^{2,p}_{\gamma}$, where $\gamma \in (2-n, 0)$ and $p > n$.  
Assume that $\sigma \in W^{1,2p}_{\gamma/2-1}$ and $\rho \in L^{p}_{\gamma-2}$ are given.  
Let $x_0 = (\phi_0,W_0)$ denote the solution to $\mathcal{F}(x,y) =0$ when $y_0 = (\tau_0,{\bf J}_0) = (0,{\bf 0})$.
Then there exists a neighborhood
of $(\phi_0,W_0)$ in which solutions to $\mathcal{F}(x,y) = 0$ exist and are unique.  In particular, there exist
unique solutions $(\phi,W)$ in this neighborhood where $(\phi-\omega,W) \in Y$ satisfies the marginally trapped surface conditions.
\end{corollary}
\begin{proof}
The existence and uniqueness of $ (\phi_0,W_0)$ such that $\phi_0 - \omega \in W^{2,p}_{\gamma}$ follows from Section 8 in \cite{yCBjIjY00}.
By Proposition~\ref{prop1:14oct13} the assumption that $g \in \cY^+$ also implies that $ \mathcal{F}'_y(x_0,y_0)$ in \eqref{eq3:20dec13}
is invertible.  Therefore, we may apply the Implicit Function Theorem to uniquely parametrize $(\phi,W) \in X$ in terms
of $(\tau,J)$ in a neighborhood of $(\phi_0,W_0)$.  Those solutions in a neighborhood of $(\phi_0,W_0)$ which correspond to $\tau \geqs 0$
will satisfy the marginally trapped surface conditions.
\end{proof}

\appendix

\section{Solutions to Semilinear, Boundary Value problems}\label{sec:app}


Suppose that $(\cM,g)$ is an $n$-dimensional,
asymptotically Euclidean manifold with boundary $\Sigma$ of class $W^{2,p}_{\gamma}$, with $\gamma \in (2-n,0)$ and $p>n$.
Denote the ends of $\cM$ by $E_i$ for $1 \le i \le m$.
Here we investigate the existence of solutions to the semilinear, Robin problem
\begin{align}\label{eq1:22dec13}
-&\Delta u = f_1(x,u) \quad \text{on $\cM$},\\
&\partial_{\nu} u = f_2(x,u) \quad \text{on $\partial\cM$.} \nonumber
\end{align} 
The functions $f_i(x,y):\cM\times I_i \to \mathbb{R}$ for $i \in\{1,2\}$ are of the form
$$
f_i(x,y) = \sum_{j=1}^{N_i} a_{ij}(x)b_{ij}(y),
$$
where each $b_{ij}(y)$ is a smooth function on $I_i\subset \mathbb{R}$,
and $a_{ij}(x) \in L^{p}_{\gamma-2}$.

In order to develop an iterative method which solves \eqref{eq1:22dec13}, we will require the following version of the
weak maximum principle.

\begin{lemma}\label{wmaxprinc}
Suppose that $(\cM,g)$ satisfies the assumptions above, and that $V(x) \in L^p_{\gamma-2}$ and
$\mu(x) \in W^{1-\frac1p,p}(x)$ are nonnegative.  If
\begin{align*}
-&\Delta u + V(x) u \ge 0 \quad \text{on $\cM$},\\
&\partial_{\nu} u + \mu(x) u \ge 0 \quad \text{on $\partial\cM$}, 
\end{align*}
and $u \to A_i \ge 0$ on each end $E_i$, then $u \ge 0$.
\end{lemma}
\begin{proof}
Let $w = -u$.  Given that $u \to A_i \ge 0$ on each end, the function $v =(w-\e)^+ $ has compact support.
By Sobolev embedding $v \in W^{1,2}$, and $wv \ge 0$. We have
\begin{align*}
\|\nabla v\|^2_{L^2(\cM )} &= \int_{\mathcal{M} } \nabla w \cdot \nabla v ~dV = -\int_{\mathcal{M} }(\Delta w) v~dV + \int_{\partial\cM }(\partial_{\nu}w) v~dA \\
             & \le -\int_{\mathcal{M} } V(x)wv~dV - \int_{\partial\cM } \mu(x)wv~dA \le 0. 
\end{align*} 
Therefore $v \equiv 0$ and $u \ge -\e$ on $\cM$.  Letting $\e \to 0$ we have that $u \ge 0$.  
\end{proof}

We also require a version of the strong maximum taken from \cite{dM05a}.  For completeness, we
state it here without proof.
\begin{lemma}\label{smaxprinc}
Suppose that $(\cM,g)$ satisfies the assumptions above and $V(x) \in L^p_{\gamma-2}$ and
$\mu(x) \in W^{1-\frac1p,p}(x)$.  Suppose $u(x) \in W^{2,p}_{\gamma}$ is nonnegative and
\begin{align*}
-&\Delta u + V(x) u \ge 0 \quad \text{on $\cM$},\\
&\partial_{\nu} u + \mu(x) u \ge 0 \quad \text{on $\partial\cM$}. 
\end{align*}
If $u(x) = 0$ for some $x \in \cM$, then $u$ vanishes identically.
\end{lemma}

In the following Lemma we construct an auxiliary, harmonic function
which allows us to freely specify the asymptotic limit $A_j$ on each end $E_j$ of
the solution to \eqref{eq1:22dec13}.
This argument is a modification of an argument given in \cite{DIMM13}.

\begin{proposition}\label{prop1:23dec13}
Suppose $\cM$ has ends $E_1,\cdots,E_k$, and let $A_j \in (-\infty,\infty)$ for
$1\le j \le k$.  Then there exists a unique function $\omega$ solving
\begin{align*}
-&\Delta \omega = 0 \quad \text{on $\cM$},\\
 &\partial_{\nu} \omega = 0 \quad \text{on $\Sigma$},
\end{align*}
which tends to $A_j$ on each end $E_j$.  Moreover,
$\min A_j \le \omega \le \max A_j$.
\end{proposition}
\begin{proof}
Let $\omega_1 = \sum \chi_j A_j$, where each $\chi_j$ is a cutoff
function which equals $1$ on the end $E_j$ and is zero outside a neighborhood of this end. 
Then $\Delta \omega_1 \in L^p_{\delta-2}$ and $\partial_{\nu}\omega_1 \in W^{1-\frac1p,p}$, so there
exists a function $\omega_2 \in W^{2,p}_\delta$ such that $\Delta \omega_2 = \Delta \omega_1$ and $\partial_{\nu} \omega_2 = \partial_{\nu}\omega_1$.
Therefore $\omega = \omega_1- \omega_2$ has the desired properties.   The fact that $\omega$ is unique follows 
by assuming that two such functions exist and showing that the difference must be identically zero. 

To show that $\min A_j \le \omega \le \max A_j$, we first pick $ \e < \min A_j$ and define $v = (\omega-\e)^- \in W^{1,2}$.
Given the asymptotic behavior of $\omega$, the function $v$ has compact support.  Therefore,
\begin{align*}
\|\nabla v\|^2_{L^2} = \int_{\cM} \nabla \omega \nabla v = 0,
\end{align*}  
which implies that $v \equiv 0$. Therefore $\omega \ge \e$, and by letting $\e \to \min A_j$ we have
$\omega \ge \min A_j$.  To show that $\omega \le \max A_j$, we make a similar argument using
$v = (\omega - \e)^+$ for $\e > \max A_j$ and let $\e \to \max A_j$.
\end{proof}

With these results in hand, we are now ready to address the existence of solutions to
\eqref{eq1:22dec13}.
The proof of the following theorem provides an iterative method to construct solutions
to this problem given sub-and supersolutions, where we recall that a sub-solution $\phi_-$
for \eqref{eq1:22dec13} satisfies
\begin{align*}
-&\Delta \phi_- \le f_1(x,\phi_-), \\
&\partial_{\nu}\phi_- \le f_2(x,\phi_-) ,
\end{align*}
and a supersolutions $\phi_+$ satisfies
\begin{align*}
-&\Delta \phi_+ \ge f(x,\phi_+), \\
&\partial_{\nu}\phi_+ \ge f_2(x,\phi_+). 
\end{align*}
Our argument is based on
the construction in \cite{yCBjIjY00}.

\begin{theorem}\label{thm1:23jan14}
Suppose that \eqref{eq1:22dec13} admits a subsolution and supersolution $\phi_-, \phi_+ \in W^{2,p}_{\gamma}$, and assume that
$$
\ell \le \phi_- \le \phi_+ \le m, \quad [l,m] \subset I_1 \cap I_2,
$$
and 
$$
\lim_{|x| \to \infty} \phi_- = \alpha, \quad \lim_{|x|\to \infty}\phi_+ = \beta.
$$
Let $\omega$ be as in Proposition~\ref{prop1:23dec13}, where each
$A_j$ satisfies $\alpha \le A_j \le \beta$.
Then Eq.\eqref{eq1:22dec13} admits a solution $\phi$ such that
$$
\phi_- \le \phi \le \phi_+, \quad \phi - \omega \in W^{2,p}_{\gamma}.
$$
\end{theorem}
\begin{proof}
As in \cite{yCBjIjY00}, the proof is by induction starting with $\phi_-$.  
Let $k_1 \in L^p_{\gamma-2}$ and $k_2 \in W^{1-\frac1p,p}$ be positive functions such that 
$$
k_1(x) \ge \sup_{l \le y \le m} \frac{\partial}{\partial y}f_1(x,y) , \quad \text{and} \quad k_2(x) \ge \sup_{l \le y \le m}  \frac{\partial}{\partial y}f_2(x,y). 
$$

We recall that by Proposition~\ref{prop1:23dec13}, $\min A_j \le \omega \le \max A_j$.
Setting $\phi_1 = \omega + u_1$, where $u_1$ satisfies
\begin{align}
-&\Delta u_1 + k_1u_1 = f_1(x,\phi_-) + k_1(\phi_- - \omega) \\
&\partial_{\nu}u_1 + k_2u_1 = f_2(x,\phi_-)+k_2(\phi_- -\omega) \nonumber, 
\end{align}
we conclude that
\begin{align*}
-&\Delta (\phi_1-\phi_-) + k_1(\phi_1-\phi_-)  \ge 0 \\
&\partial_{\nu}(\phi_1-\phi_-) + k_2(\phi_1-\phi_-)  \ge 0 \nonumber. 
\end{align*}
By assumption, $\phi_1 - \phi_- $ tends to $ A_i-\alpha \ge 0$ on each end $E_i$ and Lemma~\ref{wmaxprinc} implies
that
$$\phi_1 \ge \phi_- \quad \text{on $\cM$}.$$
Similarly,
\begin{align}\label{eq2:22dec13}
-&\Delta (\phi_+ -\phi_1) + k_1(\phi_+ - \phi_1) \ge f_1(x,\phi_+)-f_1(x,\phi_-)+k_1(\phi_+-\phi_-), \\
 &\partial_{\nu}(\phi_+-\phi_1)+k_2(\phi_+-\phi_1) \ge f_2(x,\phi_+)-f_2(x,\phi_-) + k_2(\phi_+-\phi_-). \nonumber
\end{align}
By our choice of $k_i(x)$, the function
\begin{align}\label{eq3:22dec13}
g_i(x,y) = f_i(x,y) + k_i(x)y
\end{align}
is monotonic increasing in the variable $y$.
Therefore
both equations in \eqref{eq2:22dec13} are nonnegative. Because $\phi_+ - \phi_1$ tends to
$\beta - A_i \ge 0$ on each end $E_i$, we may apply the
maximum principle~\ref{wmaxprinc} again to conclude that $\phi_1 \le \phi_+$.

We now define $\phi_n = \omega+u_n$ inductively by letting
\begin{align*}
-&\Delta u_n + k_1u_n = f_1(x,\phi_{n-1}) + k_1u_{n-1} \\
&\partial_{\nu}u_n + k_2u_n = f_2(x,\phi_{n-1})+k_2u_{n-1}  \nonumber. 
\end{align*} 
Standard elliptic theory implies that $u_n \in W^{2,p}_{\gamma}$ for each $n$ and $u_n \to \omega$ on each end.
Assume that for all $1 \le i\le k \le n-1$, $\phi_i$ and $\phi_k$ are defined as above and satisfy $\phi_- \le \phi_i \le \phi_k \le \phi_+$.
Then we have
\begin{align*}
-&\Delta (\phi_n -\phi_{n-1}) + k_1(\phi_n - \phi_{n-1}) = f_1(x,\phi_{n-1})-f_1(x,\phi_{n-2})+k_1(\phi_{n-1}-\phi_{n-2}) \ge 0, \\
 &\partial_{\nu}(\phi_n-\phi_{n-1})+k_2(\phi_n-\phi_{n-1}) = f_2(x,\phi_{n-1})-f_2(x,\phi_{n-2}) + k_2(\phi_{n-1}-\phi_{n-2}) \ge 0, \nonumber
\end{align*}
where the above inequalities follow from the inductive hypothesis and the fact that \eqref{eq3:22dec13} is monotonic increasing.
As $\phi_n - \phi_{n-1} \to 0$ on each end $E_i$, the maximum principle implies that $\phi_n \ge \phi_{n-1}$.  Finally, an application of the maximum principle to
\begin{align*}
-&\Delta (\phi_+ -\phi_{n}) + k_1(\phi_+ - \phi_{n}) \ge f_1(x,\phi_+)-f_1(x,\phi_{n-1})+k_1(\phi_+-\phi_{n-1}) \ge 0, \\
 &\partial_{\nu}(\phi_+-\phi_{n})+k_2(\phi_+-\phi_{n}) \ge f_2(x,\phi_{+})-f_2(x,\phi_{n-1}) + k_2(\phi_{+}-\phi_{n-1}) \ge 0, \nonumber
\end{align*}
implies that $\phi_n \le \phi_+$.  

Therefore, the sequence of functions $\phi_n \in W^{2,p}_{\gamma}$ is monotonic increasing and bounded above by $\phi_+$.  Thus the sequence converges to a function
$\phi(x) = \omega + u(x)$, with $\phi_- \le \phi \le \phi_+$.  A standard bootstrapping argument as in \cite{yCBjIjY00} then implies that $\phi(x)$ has the desired regularity. 
\end{proof}

\begin{theorem}\label{uniq}
Let $(\cM,g)$ be an $n$-dimensional,
asymptotically Euclidean manifold with boundary $\Sigma$ of class $W^{2,p}_{\gamma}$, where $\gamma \in (2-n,0)$ and $p > n$.
Let $\tau \in W^{1,p}_{\gamma-1}, {\sigma\in L^{2p}_{\gamma-1}}, \rho \in L^p_{\gamma-2}$, $\theta_- \in W^{1-\frac1p,p}(\Sigma)$ be fixed data such that $((n-1)\tau+|\theta_-|) \ge 0$ on $\Sigma$, and suppose that $W \in W^{2,p}_{\delta}$ is
given, where $2-n < \delta \leqs \gamma/2$ and $S(\nu,\nu) \geqs 0$.  Additionally assume that $A_i \in (0,\infty)$ and that $\omega$ is the associated smooth, harmonic function such that $\omega \to A_i$ on each end $E_i$.
Finally, assume that the Lichnerowicz equation \eqref{eq3:27nov13} with boundary conditions \eqref{eq5:11july13}
has a sub- and supersolution $\phi_-$ and $\phi_+$ which asymptotically bound $\omega$.  Then there exists a unique solution $\phi>0$ 
to the Lichnerowicz equation \eqref{eq3:27nov13} with boundary conditions \eqref{eq5:11july13} such that $\phi-\omega \in W^{2,p}_{\gamma}$.
\end{theorem}
\begin{proof}
The fact that a solution exists follows from Theorem~\ref{thm1:23jan14}.  To see that this solution is unique, suppose that $\phi_1$ and $\phi_2$ are both solutions.
That is, suppose that for $i \in\{1,2\}$
\begin{align*}
-\Delta \phi_i + c_n R \phi_i + b_n \tau^2\phi_i^{N-1}-c_n|\sigma + \cL W|^2\phi_i^{-N-1} - &c_n \rho\phi_i^{-\frac{N}{2}} = 0 \quad \text{on $\cM$}, \\
\partial_{\nu}\phi_i+d_nH\phi_i+\left(d_n \tau - \frac{d_n}{n-1}\theta_- \right)\phi_i^{\frac{N}{2}}-\frac{d_n}{n-1}S(&\nu,\nu)\phi_i^{-\frac{N}{2}} = 0 \quad \text{on $\Sigma$}. 
\end{align*}
The conformal transformation properties of the scalar curvature and boundary, mean extrinsic curvature then imply that
\begin{align*}
c_nR(\phi_i^{N-2}g) \phi_i^{N-1} = -b_n \tau^2\phi_i^{N-1}+c_n|\sigma + \cL W|^2\phi_i^{-N-1} + &c_n \rho\phi_i^{-\frac{N}{2}}, \\
d_n H(\phi_i^{N-2}g) \phi_i^{\frac{N}{2}} = -\left(d_n \tau - \frac{d_n}{n-1}\theta_- \right)\phi_i^{\frac{N}{2}}+\frac{d_n}{n-1}S(&\nu,\nu)\phi_i^{-\frac{N}{2}},
\end{align*}
where $R(\phi_i^{N-2}g)$ and $H(\phi_i^{N-2}g)$ denote the scalar curvature and boundary mean curvature with
respect to the metric $\phi_i^{N-2}g$. 
Setting $u = \phi_1^{-1}\phi_2$, we clearly have that $u -1 \in W^{2,p}_{\gamma}$.  Moreover, the above equation implies that
\begin{align}\label{eq1:26jan14}
-&\Delta_{\phi_1^{N-2}g}u + c_n(-b_n \tau^2+c_n|\sigma + \cL W|^2\phi_1^{-2N} + c_n \rho\phi_1^{-\frac{3N}{2}+1})u = \\
&\quad \quad \quad \quad \quad \quad \quad \quad \quad \quad\quad \quad \quad c_n( -b_n \tau^2+c_n|\sigma + \cL W|^2\phi_2^{-2N} + c_n \rho\phi_2^{-\frac{3N}{2}+1} )u^{N-1} ,\nonumber \\
&\partial_{\nu}u+ d_n\left(- \left(d_n \tau - \frac{d_n}{n-1}\theta_- \right)+\frac{d_n}{n-1}S(\nu,\nu)\phi_1^{-N}   \right)u = \nonumber \\
&\quad \quad \quad \quad \quad \quad \quad \quad \quad \quad\quad \quad \quad d_n\left(- \left(d_n \tau - \frac{d_n}{n-1}\theta_- \right)+\frac{d_n}{n-1}S(\nu,\nu)\phi_2^{-N}  \right) u^{\frac{N}{2}}, \nonumber
\end{align}
where $\partial_{\nu}$ is with respect to $\phi_1^{N-2}g$.  We note that the above equations have the form
\begin{align*}
-\Delta_{\phi_1^{N-2}g}u +(a + b \phi_1^{-2N} + c \phi_1^{-\frac{3N}{2}+1})u&= (a + b \phi_2^{-2N} + c \phi_2^{-\frac{3N}{2}+1})u^{N-1} \quad \text{on $\cM$},\\
\partial_{\nu} u + (e + f\phi_1^{-N})u &= (e + f \phi_2^{-N})u^{\frac{N}{2}} \quad \text{on $\Sigma$},
\end{align*}
where $a \le 0, b\ge 0, c\ge 0$ and $e \le 0, f \ge 0$.  Rearranging these two equations, we obtain
\begin{align*}
-\Delta_{\phi_1^{N-2}g}u &+ \frac{1}{(u-1)}\Big(a(1-u^{N-2})\big.\\
&+\left.b\phi_1^{-2N}(1-u^{-N-2})+c\phi_1^{-\frac{3N}{2}+1}(1-u^{-\frac{N}{2}-1})\right)u(u-1) = 0 , \\
\partial_{\nu} u &+ \frac{1}{(u-1)}\left(e(1-u^{\frac{N}{2}-1})+f\phi_1^{-N}(1-u^{-\frac{N}{2}-1})\right)u(u-1) = 0.
\end{align*}
We observe that for $m > 0$, $\frac{u^m-1}{(u-1)} > 0$ given that $u > 0$.  This also implies that $\frac{(1-u^{-m})}{(u-1)} = u^{-m}\frac{u^m-1}{(u-1)} > 0$.
Given the assumptions on $a,b,c, e$ and $f$, we conclude that
\begin{align*}
\frac{1}{(u-1)}\left(a(1-u^{N-2})+b\phi_1^{-2N}(1-u^{-N-2})+c\phi_1^{-\frac{3N}{2}+1}(1-u^{-\frac{N}{2}-1})\right)u \geqs 0 ~~~\text{on $\cM$}, \\
\frac{1}{(u-1)}\left(e(1-u^{\frac{N}{2}-1})+f\phi_1^{-N}(1-u^{-\frac{N}{2}-1})\right)u \geqs 0 ~~~\text{on $\Sigma$}.
\end{align*}
The fact that $u -1 \in W^{2,p}_{\gamma}$ and Lemma~\ref{wmaxprinc} imply that $u-1 \ge 0$ on $\cM$.  Thus, $\phi_2 \ge \phi_1$.  We may obtain the inequality $\phi_1 \ge \phi_2$
by reversing the roles of $\phi_1$ and $\phi_2$ in the above argument.  Thus, $\phi_1 = \phi_2$ and the solution is unique.
\end{proof}

\bibliographystyle{abbrv}
\bibliography{Caleb,Caleb4,Caleb5-NEW,mjh,ref-gn,books,papers}


\end{document}